\documentclass[aps,pra,10pt,twocolumn,notitlepage,showpacs,showkeys,preprintnumbers,amsmath,amssymb,nofootinbib]{revtex4-2}

\usepackage[margin=0.75in]{geometry}

\usepackage{physics}
\usepackage{amsmath}
\usepackage{amsfonts}
\usepackage{amssymb}
\usepackage{mathtools}
\usepackage{upgreek}

\usepackage{graphicx}
\usepackage{tikz}
\definecolor{myViolet}{RGB}{153, 50, 204}
\usepackage{quantikz}
\usepackage{pgfplots}
\usetikzlibrary{angles,
                graphs,
                graphs.standard, 
                fit,
                decorations.pathmorphing,
                decorations.pathreplacing,
                decorations.shapes,
                patterns, 
                patterns.meta,
                shadows,
                shapes.geometric,
                positioning}
\tikzset{snake it/.style={decorate, decoration=snake}}
\pgfplotsset{compat=1.17} 
\usepackage{tabularx}
\newcolumntype{Y}{>{\centering\arraybackslash}X}
\usepackage{makecell}

\usepackage{lmodern}

\usepackage{url}
\usepackage{microtype}
\usepackage[sort&compress]{natbib}
\usepackage{hyperref}
\hypersetup{
    colorlinks,
    linkcolor=[RGB]{153, 50, 204},
    citecolor=[RGB]{153, 50, 204},
    urlcolor=[RGB]{153, 50, 204}
}
\usepackage[capitalise]{cleveref}

\usepackage{newtxtext}
\makeatletter
\input{t1ntxtlf.fd}
\DeclareFontShape{T1}{ntxtlf}{m}{sl}{<-> \ntx@scaled ptmro8t}{}
\DeclareFontShape{T1}{ntxtlf}{b}{sl}{<-> \ntx@scaled ptmbo8t}{}
\DeclareFontShape{T1}{ntxtlf}{bx}{sl}{<-> ssub * ntxtlf/b/sl}{}
\makeatother

\usepackage{algorithm2e}
\usepackage{algpseudocode}
\RestyleAlgo{ruled}
\SetKwComment{Comment}{/* }{ */}
\SetKwInOut{Input}{input}\SetKwInOut{Output}{output}

\usepackage{amsthm}
\usepackage{thmtools}
\usepackage{thm-restate}
\newtheorem{theorem}{Theorem}
\newtheorem*{theorem*}{Theorem}
\newtheorem{result}{Result}
\newtheorem*{result*}{Result}

\newtheorem{conjecture}{Conjecture}
\newtheorem{corollary}{Corollary}
\newtheorem{lemma}{Lemma}
\newtheorem{proposition}{Proposition}
\theoremstyle{definition}
\newtheorem{definition}{Definition}
\theoremstyle{remark}
\newtheorem{remark}{Remark}

\usepackage{paralist}
\usepackage{multirow}

\newcommand{\refcite}[1]{Ref.~\cite{#1}}

\newcommand{\B}{\{0,1\}}
\newcommand{\poly}[1]{{\rm poly}\!\left(#1\right)}
\newcommand{\cl}[1]{\textnormal{{\bf #1}}}
\newcommand{\clw}[2]{\textnormal{{\bf #1}{\rm -}{\rm #2}}}
\newcommand{\clsb}[2]{\textnormal{{\bf #1}\textsubscript{#2}}}
\newcommand{\Gate}[1]{{\fontfamily{cmr}\selectfont\textup{\textsc{#1}}}}
\newcommand{\g}{\textsl{g}}
\renewcommand{\sc}[1]{\textnormal{\textsc{#1}}}
\renewcommand{\norm}[1]{\lvert\hspace{-0.1em}\lvert {#1}\rvert\hspace{-0.1em}\rvert}
\newcommand{\ind}{\boldsymbol{1}}
\newcommand{\e}{\mathrm{e}}

\begin{document}
\title{The Guided Local Hamiltonian Problem for Stoquastic Hamiltonians}
\thanks{Note that this problem is defined in the context of Ref.~\cite{gharibian2023dequantizing} and not Ref.~\cite{bravyi2015monte}.}
\author{Gabriel Waite}
\email{gabriel.waite@student.uts.edu.au}
\affiliation{Centre for Quantum Computation and Communication Technology,%
Centre for Quantum Software and Information, School of Computer Science,%
Faculty of Engineering and Information Technology,%
University of Technology Sydney, NSW 2007, Australia}
\begin{abstract}
    We show that the \textsc{Guided Local Hamiltonian} problem for stoquastic Hamiltonians is (promise) \textbf{BPP}-hard.
    The \textsc{Guided Local Hamiltonian} problem extends the \textsc{Local Hamiltonian} problem by incorporating an additional input known as a guiding state, which is promised to overlap with the ground state.
    For a range of local Hamiltonian families, prior work shows this problem is (promise) \textbf{BQP}-hard, though for stoquastic Hamiltonians, the complexity was previously unknown.
    We obtain our results by first reducing from quantum-inspired \textbf{BPP} circuits to $6$-local stoquastic Hamiltonians.
    We prove particular classes of quantum states, known as semi-classical encoded subset states, can guide the estimation of the ground-state energy.
    Our analysis shows that this \textbf{BPP}-hardness does not depend on locality, i.e., the result holds for $2$-local stoquastic Hamiltonians.
    Additional arguments extend this \textbf{BPP}-hardness to Hamiltonians restricted to a square lattice.
    We further show that for stoquastic Hamiltonians with a fixed local constraint on a subset of the system qubits, the \textsc{Guided Local Hamiltonian} problem is \textbf{BQP}-hard.
    In addition to these hardness results, we present a deterministic classical approximation algorithm for the problem under the conditions of constant promise gap, constant overlap, and constant spectral gap, when the guiding state is preparable in constant depth by a geometrically local circuit.
\end{abstract}
\keywords{Guided Local Hamiltonian, Stoquastic Hamiltonian, BPP-hardness, BQP-hardness}
\maketitle
\section{Introduction}
The \sc{Guided Local Hamiltonian} problem augments the standard \sc{Local Hamiltonian}~\cite{kitaev2002classical} problem by including a \emph{guiding state}: a quantum state promised to have non-trivial overlap with the ground state, specified by a succinct classical description.
The goal remains to decide whether the ground-state energy lies below a threshold $a$ or above a threshold $b$, promised that one of these cases holds, with the guiding state available as an additional computational resource (see \cref{def:GLHP} for a formal statement).\footnote{Equivalently, the problem can be framed as estimating the ground-state energy to inverse-polynomial precision using the guiding state.}
Including a guiding state can significantly alter the problem's computational complexity, in some cases allowing efficient decidability on a quantum computer.
In particular, the \sc{Guided $2$-Local Hamiltonian} problem is \clw{BQP}{hard}~\cite{richter2007two,gharibian2023dequantizing,cade2023improved} and several physically-relevant $2$-local Hamiltonians, including the antiferromagnetic $XY$ and Heisenberg models, retain this \clw{BQP}{hardness} in the guided setting~\cite{cade2023improved}.
Several open questions remain, including whether efficient methods exist for finding \emph{good} guiding states? 
Are there other families of Hamiltonians that are \clw{BQP}{hard}?
More optimistically, might there exist restricted families of Hamiltonians for which good guiding states can be efficiently computed~\cite{liu2021stoqma, jiang2025local, waite2025complexityb}, and if so, what is the resulting complexity of such families?

\citet{cade2023improved} examined a range of $2$-local Hamiltonian families for the \sc{Guided Local Hamiltonian} problem, including real, translation-invariant, and antiferromagnetic models; however, the class of stoquastic Hamiltonians remains unexplored.
The stoquastic case is particularly intriguing because these Hamiltonians, which have real and non-positive off-diagonal elements in the computational basis, admit sign-problem-free representations enabling classical simulation via probabilistic and Monte Carlo methods~\cite{bravyi2015monte}.
In the standard \sc{Local Hamiltonian} problem, restricting to stoquastic Hamiltonians yields the \sc{Stoquastic Local Hamiltonian}, which is \clw{StoqMA}{complete}~\cite{bravyi2006merlin,cade2023improved,bravyi2016complexity, waite2025complexitya}.
The complexity class \cl{StoqMA} captures promise problems verifiable by restricted quantum computations with specific measurement power, and is known to satisfy the containment $\cl{StoqMA} \subseteq \cl{QMA}$.
This separation raises the natural question of whether the guided version behaves differently for stoquastic Hamiltonians, and whether it yields a complexity classification distinct from the previously established \clw{BQP}{hard} cases.
More broadly, structural results for the standard \sc{Local Hamiltonian} problem suggest that such distinctions may be fundamental.
In particular, the tetrachotomy theorem of Cubitt and Montanaro~\cite[Theorem 7]{cubitt2016complexity} shows that every $2$-local interaction set $\mathcal{S}$ yields a \sc{Local Hamiltonian} problem that is either \cl{P}-, \cl{NP}-, \cl{StoqMA}-, or \cl{QMA}-complete.
This suggests the possibility of an analogous categorisation for the \sc{Guided Local Hamiltonian} problem.
In this work, we address this question by analysing general local stoquastic Hamiltonians.
We deliberately exclude the simultaneously locally diagonalisable (SLD) subclass --- whose interactions capture a \cl{StoqMA}-complete case in the tetrachotomy --- from our analysis, leaving this structured setting for future work.

Building on arguments from \citet{bravyi2006complexity}, we show that \cl{BPP} circuits can be mapped to $6$-local stoquastic Hamiltonians.
Our efficient reduction relies on the well-known circuit-to-Hamiltonian construction~\cite{kitaev2002classical} adapted for stoquastic Hamiltonians~\cite{bravyi2006complexity}.
We then construct a family of guiding states $\ket{\zeta}$, referred to as \emph{semi-classical encoded subset states} (see \cref{def:semi-classical-encoded-subset-state} for a definition), which provably have sufficient overlap with the ground state of the resulting Hamiltonians.
We conclude that the \sc{Guided $6$-Local Hamiltonian} problem for stoquastic Hamiltonians is \clw{BPP}{hard}.
Moreover, provided with a classical description of an appropriate guiding state, the problem of estimating the ground-state energy of a stoquastic Hamiltonian is at least as hard as any problem efficiently solvable by a probabilistic classical algorithm.
Consequently, no deterministic classical algorithm is expected to solve this problem unless $\cl{BPP} = \cl{P}$.\footnote{The equality $\cl{BPP} = \cl{P}$ is widely conjectured in complexity theory, supported by the existence of derandomisation techniques~\cite{arora2009computational}.}
We further show that this complexity is not an artefact of locality: the \clw{BPP}{hardness} persists under locality (perturbative gadget) reductions~\cite{waite2025complexitya} while using the same family of guiding states.
When imposing additional structural restrictions --- such as fixing a subset of system qubits to a prescribed state $\ket{\phi}$, a procedure known as \emph{pinning}~\cite{nagaj2021pinned} --- the problem becomes \clw{BQP}{hard}.
This follows from a recent embedding technique~\cite{bravyi2023rapidly} that maps general local Hamiltonians into higher-dimensional stoquastic Hamiltonians.

The \sc{Guided Local Hamiltonian} problem provides valuable theoretical insight into computational complexity boundaries, although realising the problem in physically realistic and practical settings remains challenging.
Surprisingly, when the precision required to estimate the ground-state energy is polynomially-small and the guiding state has high fidelity with the ground state, the problem remains \clw{BQP}{hard}, rather than becoming classically tractable.
So far, classical tractability has only been demonstrated under relaxed conditions, such as constant precision and overlap~\cite{gharibian2023dequantizing,zhang2024dequantized}.
It is plausible that classical algorithms exist for solving the problem to inverse-polynomial precision under additional structural or promise assumptions.
However, the class of Hamiltonians for which this holds is likely narrow and may lack physical relevance.
Despite similar bottlenecks, studying the guided problem for stoquastic Hamiltonians provides insight into the classical limits of the \sc{Guided Local Hamiltonian} framework.
Stoquastic Hamiltonians possess features amenable to classical techniques like Markov chains, yet occupy complexity classes that are not fully understood~\cite{aharonov2025stoqma}.

\subparagraph{Outline and Contributions.} 
The goal of this work parallels that of Ref.~\cite{cade2023improved}: to broaden the scope of the \sc{Guided Local Hamiltonian} problem beyond complex and real Hamiltonians.
We employ techniques such as the circuit-to-Hamiltonian construction, the Schrieffer-Wolff transformation, and perturbation gadgets.
Our focus is on a previously unexplored region of the problem --- when the Hamiltonian is stoquastic.
We show this variant is among the hardest problems in the classical complexity class (promise) \cl{BPP} (see~\cref{def:BPP}).
Our results begin with a mapping from \cl{BPP} circuits to hybrid quantum-classical circuits we call \clsb{BPP}{q} circuits (see~\cref{def:BPPq}).
We prove the existence of a semi-classical encoded subset state (see~\cref{def:semi-classical-encoded-subset-state}) that overlaps with a portion of the history state, i.e., the ground state of the Hamiltonian arising from the circuit reduction.
This follows from the use of classically reversible gates in \clsb{BPP}{q} circuits (see~\cref{def:CRQVC}).
Our main result establishes the \clw{BPP}{hardness} of the \sc{Guided Local Hamiltonian} problem for stoquastic Hamiltonians, cf.~\cref{thrm:G6LSH-BPP-hard}.
Using perturbation gadgets, we reduce the Hamiltonian's locality from six to two, cf.~\cref{thrm:G2LSH-BPP-hard}.
By applying new perturbation gadgets~\cite{waite2025complexitya}, we show that \clw{BPP}{hardness} holds even under geometrical constraints, cf.~\cref{thrm:GsquareLSH-BPP-hard}.

We complement the \clw{BPP}{hardness} results with a \clw{BQP}{hardness} result by considering the \sc{Guided} extension to the \sc{Pinned Local Hamiltonian} problem~\cite{nagaj2021pinned}.
This variant fixes a state $\ket{\phi}$ on a subset of $p$ qubits and asks whether there exists a state $\ket{\psi}$ on the remaining $n - p$ qubits such that the energy is below a threshold $a$, or whether all such states have energy above a threshold $b$.
It was previously shown that the \sc{Pinned Stoquastic $3$-Local Hamiltonian} problem is \clw{QMA}{complete}~\cite{nagaj2021pinned} via an embedding of a $2$-local general \clw{QMA}{complete} Hamiltonian into a $3$-local stoquastic one~\cite{bravyi2023rapidly}.
The problem we study, referred to as \sc{Guided Pinned Stoquastic Local Hamiltonian}, is analogous to the standard guided problem but with the additional constraint of fixed qubits in the state $\ket{\phi}$.
Using the same embedding technique, we show that this guided pinned problem is \clw{BQP}{hard}, cf.~\cref{thrm:GP2LSH-BQP-hard}.

In addition to our main results, we present a classical algorithm for the \sc{Guided Local Hamiltonian} problem that applies when the promise gap, overlap, and spectral gap are all constant, and the guiding state is preparable in constant depth by a geometrically local circuit.
This result is motivated by Ref.~\cite{zhang2024dequantized}, which relaxes the norm bound condition of Ref.~\cite{gharibian2023dequantizing} from constant to $O(\poly{n})$ and employs cluster expansion techniques~\cite{mann2024algorithmic} to estimate the ground-state energy.
We demonstrate that the cluster expansion framework of Ref.~\cite{mann2024algorithmic} can be applied directly to decide the \sc{Guided Local Hamiltonian} problem, yielding a deterministic polynomial-time algorithm under the conditions above.
Our contributions are twofold: 
\begin{inparaenum}[(i)]
    \item we provide a simplified and self-contained proof that clarifies the approach, and
    \item we eliminate the requirement for an a priori energy interval containing the ground-state energy, which is assumed as an additional promise in Ref.~\cite{zhang2024dequantized}.
\end{inparaenum}
While this classical algorithm does not apply to our hardness results, it provides useful context for understanding the boundary between classical tractability and quantum hardness in guided Hamiltonian problems.

To conclude, we present several open questions inspired by our results.
Our summary suggests the existence of an intermediate class between \cl{BPP} and \cl{BQP} capable of approximating the eigenvalues of local stoquastic Hamiltonians to inverse-polynomial precision.

\section{Preliminaries}\label{sec:preliminaries}
A $k$-local Hamiltonian over $n$ qubits is a Hermitian operator expressed as $H = \sum_{j=1}^{m} h_j$, where each $h_j$ is a Hermitian term acting non-trivially on at most $k$ qubits, and $m = \poly{n}$.
For a given Hamiltonian let $\Pi_0$ denote the projector onto the ground space of the Hamiltonian and $\lambda_0(H)$ the smallest eigenvalue of the Hamiltonian.
We denote the ground state of the Hamiltonian as $\ket{\phi_0}$.
For the purposes of this work we assume all Hamiltonians satisfy $\norm{H}\leq 1$ and have a non-degenerate ground space with a spectral gap of $\Omega(1/\poly{n})$.

A stoquastic Hamiltonian is a Hamiltonian with real non-positive off-diagonal elements in the computational basis.
Due to the Hermicity of the operator, all elements of a stoquastic matrix are real.

\subsection{Subset states}

\begin{definition}[Subset State]
    For any subset $S \subseteq \B^n$, the subset state $\ket{S}$ over $S$ is defined as
    \begin{equation}
        \ket{S} \coloneqq \frac{1}{\sqrt{\abs{S}}} \sum_{x \in S} \ket{x}.
    \end{equation}
\end{definition}

Restricting the subset to be of polynomial size gives the definition of a semi-classical subset state.
In this case, the set $S$ can be expressed with a polynomial number of bits.

\begin{definition}[Semi-Classical Subset State]
    For any subset $S \subset \B^n$ such that $\abs{S} = O(\poly{n})$, the semi-classical subset state $\ket{S}$ over $S$ is defined as
    \begin{equation}
        \ket{S} \coloneqq \frac{1}{\sqrt{\abs{S}}} \sum_{x \in S} \ket{x}.
    \end{equation}
\end{definition}

Further modification to subset states can be made by allowing for isometries acting on the state components.
Given a $n$-bit string $x$ and the associated state $\ket{x}$, we allow a set of $n$ isometries $\{V_j\}_{j\in[n]}$, where $V_j : \mathbb{C}^2 \to (\mathbb{C}^2)^{\otimes m_j}$ with $m_j = O(1)$ for each $j$, acting as 
\begin{equation*}
    \ket{x} = \bigotimes_{j\in [n]} \, \ket{x_j} \xmapsto{\{V_j\}_{j}} \bigotimes_{j\in[n]}\, V_j\ket{x_j}.
\end{equation*}

\begin{definition}[Semi-Classical Encoded Subset States]
    \label{def:semi-classical-encoded-subset-state}
    For any subset $S \subset \B^n$ such that $\abs{S} = O(\poly{n})$, let $\mathcal{V} = \{V_j\}_{j\in[n]}$ be an ordered set of isometries where, for each $j$ we have $V_j : \mathbb{C}^2 \to (\mathbb{C}^2)^{\otimes m_j}$ with $m_j = O(1)$.
    The semi-classical encoded subset state $\ket{S_{\mathcal{V}}}$ over $(S,\mathcal{V})$ is defined as
    \begin{equation}
        \ket{S_{\mathcal{V}}} \coloneqq \frac{1}{\sqrt{\abs{S}}} \sum_{x \in S} \bigotimes_{j\in[n]}\, V_j\ket{x_j}.
    \end{equation}
    The components $\ket{x} = \bigotimes_{j\in[n]}\, \ket{x_j}$ are the standard basis states.
\end{definition}

Note that when the isometries $V_j$ map a single qubit to a single qubit, i.e., $m_j = 1$, we still refer to the state as a semi-classical encoded subset state, however, each such isometry can be realised with a single qubit gate.
Moreover, it may be the case that the a subset of the isometries are the identity, i.e., $V_j = I$, therefore a semi-classical subset state is a special case of the encoded variant.
It can be argued in the case where the isometries are single qubit unitaries that a new family of states should be defined, e.g., \emph{Locally Unitarily Transformed Semi-Classical Subset States}; we do not pursue this here but note some results hold for this case.

Two simple facts follow from the definitions above.
One is that a semi-classical subset state is a special case of a semi-classical encoded subset state where the isometries are the identity.
The second is that the description of a semi-classical encoded subset state can be expressed with a polynomial number of bits.

The following relations follow from the definitions above.

\begin{restatable}[]{proposition}{tensorproductsubsetstates}\label{prop:tensor-product-subset-states}
    A tensor product of a polynomial number of semi-classical encoded subset states is a semi-classical encoded subset state.
\end{restatable}

Ref.~\cite{cade2023improved} establishes the following property of semi-classical encoded subset states.

\begin{proposition}[\textnormal{\cite[Lemma 4]{cade2023improved}}]\label{prop:sample-from-subset-state}
    Given the description of a semi-classical encoded subset state $\ket{S_{\mathcal{V}}}$, it is possible to efficiently sample from the probability distribution outputting the $M$-bit string $z \in \B^{M}$ with probability $\abs{\braket{z}{S_{\mathcal{V}}}}^2$.
\end{proposition}

Here we denote $M\coloneqq \sum_{j\in[n]} m_j$ as the size of the bit strings $x\in S$ after the action of the isometries.

\subsection{Complexity classes}
Throughout this work, we consider only \emph{promise problem} variants (unless explicitly specified otherwise), rather than any semantic definition.
We therefore drop all `promise' prefixes.
A notion of ``hard'' or ``complete'' problems for the classes we consider is therefore appropriate.
We define the complexity class \cl{BPP} using the circuit model language~\cite{watrous2008quantum}.

\begin{definition}[\cl{BPP}]\label{def:BPP}
    Let $L = (L_{\sc{yes}}, L_{\sc{no}})$ be a promise problem.
    A problem $L$ belongs to the class \cl{BPP} if there exists a polynomial-time generated Boolean circuit family $\mathcal{C} = \{C_n : n \in \mathbb{N}\}$, where each circuit $C_n$ acts on $n + r(n)$ input bits and produces one output bit, such that:
    \begin{itemize}
        \item If $x \in L_{\sc{yes}}$, then $\Pr_{y\sim U}[C_{|x|}(x,y) = {\tt 1}] \geq \frac{2}{3}$, 
        \item If $x \in L_{\sc{no}}$, then $\Pr_{y\sim U}[C_{|x|}(x,y) = {\tt 1}] \leq \frac{1}{3}$.
    \end{itemize}
    We denote $U(\{0,1\}^{r(|x|)})$, the uniform distribution over the $r(|x|)$-bit strings, as $U$.
\end{definition}

We now define the complexity class \cl{BQP}.

\begin{definition}[\cl{BQP}]\label{def:BQP}
    Let $L = (L_{\sc{yes}}, L_{\sc{no}})$ be a promise problem.
    A problem $L$ belongs to the class \cl{BQP} if there exists a polynomial-time generated quantum circuit family $\mathcal{Q} = \{Q_n : n \in \mathbb{N}\}$, where each circuit $Q_n$ acts on $n + r(n)$ input qubits and produces one output bit, such that:
    \begin{itemize}
        \item If $x \in L_{\sc{yes}}$, then $\Pr[Q_{|x|}(x) = {\tt 1}] \geq \frac{2}{3}$,
        \item If $x \in L_{\sc{no}}$, then $\Pr[Q_{|x|}(x) = {\tt 1}] \leq \frac{1}{3}$.
    \end{itemize}
\end{definition}

\subsection{Simulator Hamiltonians}
We follow the definitions of simulator Hamiltonians given in Refs.~\cite{bravyi2016complexity,cubitt2018universal}.
Consider two Hamiltonians $H_A$ and $H_B$ acting on Hilbert spaces $\mathcal{H}_A$ and $\mathcal{H}_B$, respectively.
An isometric encoding map $\mathcal{E} : \mathcal{H}_A \to \mathcal{H}_B$ and Hamiltonian $H_B$ are said to be an $(\eta,\epsilon)$-simulator for $H_A$ if there exists an isometry $\mathcal{V} : \mathcal{H}_A \to \mathcal{H}_B$ such that: the image of $\mathcal{V}$ is the low-energy subspace of $H_A$, $\norm{H_A - \mathcal{V}^\dagger H_B \mathcal{V}} \leq \epsilon$, and $\norm{\mathcal{E} - \mathcal{V}} \leq \eta$.
Equivalently, the low-energy subspace of $H_B$ approximates $H_A$, up to controlled errors.
To prove our hardness results we always take $\eta, \epsilon = O(1/\poly{n})$.
This is possible since all applications of this technique within this work can be achieved with polynomially small error parameters~\cite{bravyi2016complexity, bravyi2006complexity, waite2025complexitya}.

A consequence of this simulation is that the $j$-th smallest eigenvalues of $H_A$ and $H_B$ are close, satisfying $\norm{\lambda_j(H_B) - \lambda_j(H_A)} \leq \epsilon$, where $\lambda_j(H)$ denotes the $j$-th smallest eigenvalue of $H$~\cite[Lemma 1]{bravyi2016complexity}.
Additionally, it can be shown that the ground states of $H_A$ and $H_B$ are close under the encoding~\cite[Lemma 2]{bravyi2016complexity} --- this will be discussed later.

From a complexity-theoretic perspective, this notion of simulation ensures that any algorithm capable of solving the local Hamiltonian problem for $H_B$ can be used to solve it for $H_A$ with comparable accuracy.
This hardness preservation is formalised in the following proposition.

\begin{proposition}
    Let $H_B$ be a $(\eta,\epsilon)$-simulator for $H_A$.
    Then $H_B$ is at least as hard as $H_A$.
\end{proposition}

\begin{proof}
    Take an instance of $H_A$ to be defined as $x\coloneqq(H_A, a,b)$ such that $b-a \geq 1/\poly{n}$.
    Since $H_B$ is a $(\eta,\epsilon)$-simulator for $H_A$, we set the parameters $b' = b - \epsilon$ and $a' = a + \epsilon$.
    Setting $\epsilon < (b-a)/2$ ensures that $b'-a' = b-a - 2\epsilon \geq 1/\poly{n}$, meeting the criterion for a valid instance of $H_B$.
    It is not hard to see that in the event of a \sc{yes} case for $H_A$, i.e.
    $\lambda(H_A) \leq a$ then it must be that $\lambda(H_B) \leq a'$.
    Similarly the converse holds for the \sc{no} case.
\end{proof}

\subsection{Problem statement}
\begin{definition}[The \sc{Guided Local Hamiltonian} Problem]\label{def:GLHP}
    Given a $k$-local Hamiltonian $H$ acting on $n$ qubits such that $\norm{H} \leq 1$, parameters $a,b \in [0,1]$ such that $b-a \geq 1/\poly{n}$ and a description of a semi-classical encoded subset state $\ket{\zeta}$ with the promise that $\norm{\Pi_0\ket{\zeta}}^2 \geq \delta$ for some $0 <\delta < 1$, decide whether $\lambda_0(H) \leq a$ or $\lambda_0(H) \geq b$, promised one is true.
\end{definition}

The \sc{Guided Local Hamiltonian} problem for stoquastic Hamiltonians is defined similarly, with the additional constraint that the Hamiltonian is stoquastic.
We refer to the problem in this manner rather than ``\emph{The Guided Local Stoquastic Hamiltonian problem}'' to avoid cross-contamination with the problem presented by Bravyi in Ref.~\cite{bravyi2015monte}.
The description of the state $\ket{\zeta}$ is given in a classically efficient form, allowing for classical sample-query access~\cite{gharibian2023dequantizing} (\cref{prop:sample-from-subset-state}).
Such access is motivated by standard assumptions in classical simulation techniques~\cite{vandennest2011simulating,rudi2020approximating}.
The local Hamiltonian and parameters $a,b$ are presented using a polynomial number of bits, with the Hamiltonian encoded as a sum of local terms, each in the computational basis.
Note that in the \sc{Guided Local Stoquastic Hamiltonian} problem of Ref.~\cite{bravyi2015monte}, the guiding state is not regarded as a part of the input and furthermore, has additional criteria that is not present in our problem.

The criterion on the norm of the Hamiltonian is strong, in the sense that we would typically expect a physical Hamiltonian to have a norm bounded by a polynomial in $n$.
The renormalisation of the Hamiltonian can be achieved by a global factor scaling of $1/\poly{n}$.
The norm bound is a key factor in the applicability of dequantisation algorithms for the \sc{Guided Local Hamiltonian} problem.
Under relaxed conditions, such as constant precision and overlap, Ref.~\cite{gharibian2023dequantizing} gives a classical algorithm for problem when the norm is bounded by $1$ whereas Ref.~\cite{zhang2024dequantized} shows that the unphysical norm bound can be circumvented, allowing for a wider range of physically motivated and practical Hamiltonians to be considered.
For consistency with the original literature, we maintain the above assumptions but recognise the realistic limitations this imposes.

\section{Quantum BPP Circuits}
The complexity class \cl{BPP} captures the set of promise problems that are efficiently solvable by probabilistic algorithms~\cite{arora2009computational}.
Randomised algorithms involve random choices in their computation.
To perform this, it is sufficient to construct a random number generator producing the bit {\tt 0} or {\tt 1} with probability $1/2$, e.g., a coin flip.

A convenient way to represent such computation is via a \emph{coherent} implementation, in which randomness is encoded unitarily rather than sampled explicitly.
Translating this idea to a quantum circuit requires restricting to a subclass of quantum operations.
Specifically, we consider circuits comprised only of classically reversible gates~\cite{bennett1973logical}, i.e., from the set $\{I, X,\Gate{Cnot},\Gate{Toffoli}\}$.

The quantum circuit takes as input $\ket{x}$ for $x \in \B^n$, together with $m$ ancillae initialised in the state $\ket{0}$ and $p$ ancillae initialised in the state $\ket{+}$.
The $\ket{+}$ ancillae coherently encode the randomness of the computation, replacing classical coin flips by superposition.
The final measurement of the computation is on one qubit, in the $Z$-basis; assume without loss of generality that this is the first qubit.

We refer to this model as a \emph{(coherent) classically reversible verification circuit}.
In typical fashion we require that there is only a polynomial number of gates and a polynomial number of ancilla qubits.
This is taken for granted, i.e., the circuits are \emph{efficient}.

\begin{definition}[Classically Reversible Quantum Verification Circuit (CRQVC)]\label{def:CRQVC}
    A \emph{classically reversible quantum verification circuit} is a quantum circuit is a tuple $M_n = (n,m,p,U)$ where $n$ is the number of input qubits, $m$ is the number of ancillae initialised in the state $\ket{0}$, $p$ is the number of ancillae initialised in the state $\ket{+}$.
    The circuit $U$ is a quantum circuit on $n+m+p$ qubits composed of $K = \poly{n}$ classically reversible gates from the set $\{I, X,\Gate{Cnot},\Gate{Toffoli}\}$.
    The acceptance probability of the circuit is defined as
    \begin{equation*}
        \Pr[U(\ket{x}\ket{0^m}\ket{+^p}) = {\tt 1}] = \bra{\phi}U^\dagger \Pi_{\text{out}} U\ket{\phi},
    \end{equation*}
    where $\ket{\phi} = \ket{x,0^m,+^p}$ and $\Pi_{\text{out}} = \ketbra{1}_1$ is the projector onto the first qubit in the $Z$-basis.
\end{definition}

Note that $m,p = O(\poly{n})$.
This definition aligns with the semi-classical verification circuits defined in Ref.~\cite{waite2025complexitya}.
We have opted for the term `classically reversible quantum verification circuit' to avoid confusion with the semi-classical subset states.
The class \clsb{BPP}{q} is defined as follows:

\begin{definition}[\clsb{BPP}{q}]\label{def:BPPq}
    Let $L = (L_{\sc{yes}}, L_{\sc{no}})$ be a promise problem.
    A problem $L$ belongs to the class \clsb{BPP}{q} if there exists a polynomial-time generated CRQVC family $\mathcal{M} = \{M_n : n \in \mathbb{N}\}$, where each circuit $M_n$ acts on $n + m + p$ input bits and produces one output bit, such that:
    \begin{itemize}
        \item If $x \in L_{\sc{yes}}$, then $\Pr[M_{|x|}(x) = {\tt 1}] \geq \frac{2}{3}$,
        \item If $x \in L_{\sc{no}}$, then $\Pr[M_{|x|}(x) = {\tt 1}] \leq \frac{1}{3}$.
    \end{itemize}
\end{definition}

\begin{restatable}[]{theorem}{BPPqBPP}
    \label{thrm:BPPq-is-BPP}
    \clsb{BPP}{q} $=$ \cl{BPP}.
\end{restatable}

The proof of this theorem is outlined in~\cref{app:proofs}.
The typical amplification procedure for \cl{BPP} circuits can be applied to \clsb{BPP}{q} circuits.
This is done by repeating the computation a polynomial number of times and taking the majority output.
Moreover, \cl{BPP}$(2/3,1/3)$ = \cl{BPP}$(1 - 2^{-f(n)}, 2^{-f(n)})$ for any polynomially-bounded function $f : \mathbb{N} \to \mathbb{N}$, with $f(n) \geq 2$; we can therefore assume our statistics for \clsb{BPP}{q} are $1-2^{-f(n)}$ and $2^{-f(n)}$ respectively.
The implications of~\cref{thrm:BPPq-is-BPP} are that the Feynman-Kitaev circuit-to-Hamiltonian construction~\cite{kitaev2002classical} can easily be applied.

\subsection{Stoquastic Arthur}
A natural question regarding \cl{StoqMA} concerns the position of its verification circuit within the complexity hierarchy.
When viewing \cl{MA} and \cl{QMA} as interactive proof systems, the verifiers are limited to \cl{BPP} and \cl{BQP} computations, respectively~\cite{babai1985trading,complexityzoo}.
This follows from the fact that for promise problem complexity classes, the existential quantifiers $\exists$ (classical) and $\hat{\exists}$ (quantum) yield $\cl{MA} = \exists\cdot\cl{BPP}$ and $\cl{QMA} = \hat{\exists}\cdot\cl{BQP}$, respectively.
Contrary to expectation, \cl{StoqMA}'s verifier does not lie between these classes.
Specifically, a \emph{stoquastic Arthur} can be simulated by a \cl{BPP} circuit.
This result can be seen due to the sensitivity of \cl{StoqMA} to the proof type and completeness conditions~\cite{bravyi2006merlin,aharonov2025stoqma,liu2021stoqma}.
For instance, \cl{eStoqMA} and \cl{StoqMA}\textsubscript{$1$} are subsets of \cl{MA}~\cite{bravyi2006merlin,liu2021stoqma}.
In analogy the existential quantifier viewpoint, we can say $\cl{StoqMA} = \hat{\exists}\cdot\cl{StoqP}$ (see \cref{thm:stoqma-equals-stoqp}), where $\cl{StoqP}$ can be viewed as \cl{StoqMA} without a proof.
We formally define \cl{StoqP} in~\cref{app:StoqP} and provide a proof to justify this particular choice.

We leverage a result from Ref.~\cite{liu2021stoqma} to show that \cl{StoqP} is contained in \clsb{BPP}{q}.
\begin{result}\label{res:cStoqMA-is-MA}
    For any $1/2 \leq \beta < \alpha \leq 1$ and $\alpha- \beta \geq 1/\poly{n}$,
    \begin{equation*}
        \cl{cStoqMA}(\alpha,\beta) \subseteq \cl{MA}(2\alpha- 1,2\beta- 1).
    \end{equation*}
\end{result}

This result shows that when Merlin is restricted to classical proofs, \cl{MA} can simulate \cl{cStoqMA}.
Specifically, if $V$ is the verifier for \cl{cStoqMA}, and $M$ is the verifier for \cl{MA}, for the proof $\xi \in \{0,1\}^{w}$, then 
\begin{equation*}
    \Pr[V ~{\text{accept}} ~\xi] = \frac{1}{2} + \frac{1}{2} \Pr[M ~{\text{accept}} ~\xi].
\end{equation*}
Using our suggested definition of \cl{StoqP}, we adapt this result by replacing the classical proof $\xi$ with the instance $x$ --- treating $x$ like a `pseudo-proof'.
The same arguments imply that \cl{StoqP} $\subseteq$ \clsb{BPP}{q}.
Moreover, if $\cl{StoqP} = \cl{StoqP}(\alpha,\beta)$ for some $1/2 \leq \beta < \alpha \leq 1$ and $\alpha-\beta \geq 1/\poly{n}$, such that \cl{StoqP} relates to \cl{StoqMA} without a proof, then 
\begin{equation*}
    \cl{StoqP}(\alpha,\beta) \subseteq \clsb{BPP}{q}(2\alpha-1,2\beta-1),
\end{equation*}
follows from~\cref{res:cStoqMA-is-MA}.
\cref{app:StoqP} outlines the proof.
To summarise, the verification circuit of a sole stoquastic Arthur is no more powerful than a \cl{BPP} circuit.

\subsection{The action of classically reversible gates}
From \cref{def:BPPq}, we restrict our attention to the set of classically reversible gates $\{I, X,\Gate{Cnot},\Gate{Toffoli}\}$.
The Toffoli gate is well-known to be universal for classical reversible computation~\cite{toffoli1980reversible}.
This implies that, without loss of generality, any \clsb{BPP}{q} circuit can be expressed solely as a sequence of Toffoli gates.
The action of classical reversible gates on $n$-bit strings is well-defined and can be exactly quantified.

For a subset $S\subseteq \B^n$, let the action of a classically reversible gate $R \in \B^{2^n \times 2^n}$ be denoted $R(S)$.
Specifically, $R(S) = \{R\,x : x\in S\}$.

\begin{proposition}
    For any subset $S \subseteq \B^n$, the action of a classically reversible gate $R$ leaves the cardinality of the subset invariant, i.e., $\abs{S} = \abs{R(S)}$.
\end{proposition}

This follows from the fact that the gate $R$ is a bijective map on the set of $n$-bit strings.
The impact on subset states is then captured in the following proposition.

\begin{restatable}[]{proposition}{classicallyreversiblegateonsubsetstate}
    \label{prop:classically-reversible-gate-on-subset-state}
    Given a subset state $\ket{S}$, the action of a classically reversible gate $R$ on the state results in a subset state $\ket{R(S)}$ such that $\abs{S} = \abs{R(S)}$.
\end{restatable}

We will consider the action a sequence of classically reversible gates has on a subset state, informed by the structure of the history state that will be produced by the circuit-to-Hamiltonian construction~\cite{kitaev2002classical}.
This will be beneficial in the subsequent analysis of the \sc{Guided Local Hamiltonian} problem.

\begin{restatable}[]{lemma}{historystatesubsetstate}
    \label{lma:history-state-subset-state}
    Let $\mathcal{R} = (R_k)_{k\in [K]}$ be a sequence of $K$ classically reversible gates.
    Let $\ket{S}$ be a subset state over $S\subseteq \B^n$.
    The state 
    \begin{equation}\label{eq:history-state-one}
        \ket{\mathcal{S}} = \frac{1}{\sqrt{K}} \sum_{k \in [K]} R_k \cdots R_1 \ket{S}\ket{1^{k}\,0^{K-k}},
    \end{equation}
    is a subset state over 
    \begin{equation}
        \mathcal{S} = \bigcup_{k \in [K]} \big(\left(R_k \cdots R_1(S)\right)\times \{1^{k}\,0^{K-k}\}\big),
    \end{equation}
    where $\abs{\mathcal{S}} = \abs{S}\cdot K \leq 2^{n + \log_2(K)}$.
\end{restatable}

The proof of this lemma is given in~\cref{app:proofs}.
The state $\ket{\mathcal{S}}$ can be efficiently expressed as a polynomial list of $n+K$-bit strings.

\begin{corollary}
    \label{cor:history-state-semi-classical-subset-state}
    Let $\mathcal{R} = (R_k)_{k\in [K]}$ be a sequence of $K$ classically reversible gates such that $K = p(n)$ for some polynomial $p$.
    Let $\ket{S}$ be a semi-classical subset state over $S\subset \B^n$ with $\abs{S} = q(n)$ for some polynomial $q$.
    The state 
    \begin{equation}\label{eq:history-state-two}
        \ket{\mathcal{S}} = \frac{1}{\sqrt{K}} \sum_{k \in [K]} R_k \cdots R_1 \ket{S}\ket{1^{k}\,0^{K-k}},
    \end{equation}
    is a semi-classical subset state over 
    \begin{equation*}
        \mathcal{S} = \bigcup_{k \in [K]} \big(\left(R_k \cdots R_1(S)\right)\times \{1^{k}\,0^{K-k}\}\big),
    \end{equation*}
    where $\abs{\mathcal{S}} = \abs{S}\cdot K = q(n)\cdot p(n) \leq \poly{n}$.
\end{corollary}

The resulting subset states from~\cref{eq:history-state-one} and~\cref{eq:history-state-two} can be expressed as 
\begin{equation*}
    \ket{\mathcal{S}} = \frac{1}{\sqrt{\abs{S}\cdot K}} \sum_{s\in \mathcal{S}} \ket{s}.
\end{equation*}
It follows from~\cref{prop:sample-from-subset-state} that we can efficiently sample from the probability distribution outputting associated with~\cref{eq:history-state-two}.
Note that~\cref{lma:history-state-subset-state} cannot be straightforwardly applied to semi-classical encoded subset states unless the form of each isometry results in a mapped state of equal (uniform) amplitude, i.e., $V_j\ket{x_j}$ is itself a subset state.
The following lemma is a consequence of~\cref{lma:history-state-subset-state} and~\cref{def:semi-classical-encoded-subset-state}.

\begin{restatable}[]{lemma}{partialhistorySCESS}
    \label{lma:partial-history-SCESS}
    Let $\mathcal{R} = (R_k)_{k\in [K]}$ be a sequence of $K$ classically reversible gates such that $K = p(n)$ for some polynomial $p$.
    Furthermore, let $\mathcal{V} = \{V_j\}_{j\in[n]}$ be an ordered set of isometries such that the first $X \subseteq [n]$ isometries are the identity.

    Assume that for any $k\in [K]$, the gate $R_k$ has support only on qubits coinciding with $X$.
    Let $\ket{S_{\mathcal{V}}}$ be a semi-classical encoded subset state over $(S,\mathcal{V})$ with $S \subset \B^n$ and $\abs{S} = q(n)$ for some polynomial $q$.
    The state 
    \begin{equation*}
        \ket{\mathcal{S}_{\mathcal{V}}} = \frac{1}{\sqrt{K}} \sum_{k \in [K]} R_k \cdots R_1 \ket{S_{\mathcal{V}}}\ket{1^{k}\,0^{K-k}},
    \end{equation*}
    is a semi-classical encoded subset state over
    \begin{equation*}
        \mathcal{S} = \bigcup_{k \in [K]} \big(\left(R_k \cdots R_1(S)\right)\times \{1^{k}\,0^{K-k}\}\big),
    \end{equation*}
    where $\abs{\mathcal{S}} = \abs{S}\cdot K = q(n)\cdot p(n) \leq \poly{n}$.
\end{restatable}

The proof of this lemma is given in~\cref{app:proofs}.
The state $\ket{\mathcal{S}_{\mathcal{V}}}$ can be efficiently expressed as a polynomial list of $n+K$-bit strings paired with a sequence of $n+K$ isometries.

\section{Classical Hardness}
Starting with a \clsb{BPP}{q} circuit $M_{|x|}$ that decides the promise problem $L = (L_{\sc{yes}}, L_{\sc{no}})$, we apply the Feynman-Kitaev circuit-to-Hamiltonian construction to map the circuit to a $6$-local stoquastic Hamiltonian.
The difference here in contrast to typical reductions of this type, is: \begin{inparaenum}[(i)] \item the application of the Schrieffer-Wolff transformation to produce the final Hamiltonian of the reduction, \item no proof state and \item proving the existence of a semi-classical subset guiding state --- $\ket{\zeta}$. \end{inparaenum}

The following lemma will be useful in the subsequent analysis when bounding the difference between a triplet of states.

\begin{restatable}[]{lemma}{normtracking}
    \label{lma:norm-tracking}
    Let $\ket{a}, \ket{b}, \ket{c} \in (\mathbb{C}^2)^{\otimes n}$ be normalised states, such that $\norm{\ket{a} - \ket{b}} \leq \epsilon_{ab}$ and $\abs{\braket{b}{c}}^2 \geq \delta_{bc}$. Then:
    \begin{equation*}
        \begin{cases}
            (\sqrt{\delta_{bc}} - \epsilon_{ab})^2 \leq \abs{\braket{a}{c}}^2 \leq (\sqrt{\delta_{bc}} + \epsilon_{ab})^2 & \text{if}~ \epsilon_{ab} \leq \sqrt{\delta_{bc}}, \\[0.25cm]
            0 \leq \abs{\braket{a}{c}}^2 \leq (\sqrt{\delta_{bc}} + \epsilon_{ab})^2 & \text{if}~ \epsilon_{ab} > \sqrt{\delta_{bc}}.
            \end{cases}
    \end{equation*}
\end{restatable}

For the proof, see~\cref{app:proofs}.

\begin{restatable}[]{theorem}{GLSHBPPhard}
    \label{thrm:G6LSH-BPP-hard}
    The \sc{Guided $6$-Local Hamiltonian} problem for stoquastic Hamiltonians is \clw{BPP}{hard} for any $\delta \in (1/\poly{n}, 1 - 1/\poly{n})$.
\end{restatable}

\begin{proof}
    Let $L = (L_{\sc{yes}}, L_{\sc{no}})$ be a promise problem in \clsb{BPP}{q}.
    Define $x \in \B^n$ to be an input to the problem.
    The circuit $M_{|x|} = R_K \cdots R_1$ is a classically reversible quantum circuit that decides $x$; the total number of gates $K = p(n)$ where $n \eqqcolon |x|$.
    The circuit acts on the input $\ket{x}$ and the ancillae $\ket{0^m}\ket{+^p}$.
    If $x \in L_{\sc{yes}}$, then $\Pr[M_{|x|}(\ket{x,0^m,+^p}) = {\tt 1}] \geq \frac{2}{3}$; if $x \in L_{\sc{no}}$, then $\Pr[M_{|x|}(\ket{x,0^m,+^p}) = {\tt 1}] \leq \frac{1}{3}$.

    Now we define a $6$-local Hamiltonian as follows:
    \begin{align*}
        H &\coloneqq H_{\text{in}} + H_{\text{clock}} + H_{\text{prop}} ,\\
        H_{\text{in}} &\coloneqq \bigg(\sum_{j=1}^{n} \Pi^{(\bar{x}_j)}_j + \sum_{j = 1}^{m} \Pi^{(1)}_j + \sum_{j = 1}^{p} \Pi^{(-)}_j\bigg) \otimes \ketbra{0}{0}_{c_1}, \\
        H_{\text{clock}} &\coloneqq \sum_{t=1}^{K} \ketbra{01}{01}_{c_{t-1},c_t},\\
        H_{\text{prop}} &\coloneqq \sum_{t=1}^{K} H_{\text{prop}}^{(t)}.
    \end{align*}
    Here $\Pi^{(\bar{x}_j)}_j = I - \ketbra{x_j}{x_j}_j$, $\Pi^{(1)}_j = \ketbra{1}{1}_j$, $\Pi^{(-)}_j = \ketbra{-}{-}_j$.
    The propagation terms are defined as 
    \begin{align*}
        H_{\text{prop}}^{(1)} &\coloneqq \ketbra{00}_{c_1,c_2} + \ketbra{10}_{c_1,c_2} \\
            &\qquad- R_1\otimes\big(\ketbra{10}{00}_{c_1,c_2} + \ketbra{00}{10}_{c_1,c_2}\big),\\
        H_{\text{prop}}^{(t)} &\coloneqq \ketbra{100}_{c_{t-1},c_t,c_{t+1}} + \ketbra{110}_{c_{t-1},c_t,c_{t+1}} \\
            &\qquad - R_t\otimes\ketbra{110}{100}_{c_{t-1},c_t,c_{t+1}} \\
            &\qquad- R_t\otimes\ketbra{100}{110}_{c_{t-1},c_t,c_{t+1}}, ~ 1 < t < K\\
        H_{\text{prop}}^{(K)} &\coloneqq \ketbra{10}_{c_{T-1},c_T} + \ketbra{11}_{c_{T-1},c_T} \\
            &\qquad- R_T\otimes(\ketbra{11}{10}_{c_{T-1},c_T} + \ketbra{10}{11}_{c_{T-1},c_T}).
    \end{align*}

    $H$ is clearly $6$-local due to the structure of the clock terms coupled to Toffoli gates.
    Furthermore, it is easy to see $H$ is also stoquastic since $H_{\text{in}}$ and $H_{\text{clock}}$ are diagonal in the computational basis and the terms of the form $R_t\otimes(\dots)$ are off-diagonal but all carry a negative sign.

    \subparagraph{Pre-idling.} We show that padding the circuit $\mathcal{V}_x$ with a polynomial number of identity gates, enables us to split the history state into two components.
    For the original $K$-gate circuit, assume that the first $T = q(n)$ many gates only act non-trivially on the first $n + m$ qubits.
    We construct a new circuit $\hat{M}_{|x|}$ by padding the start of the circuit $M_{|x|}$ with $N = r(n)$ identity gates, giving a total of $\hat{K} \coloneqq K + N$ gates and the sequence $\hat{R}_{K+N} \cdots \hat{R}_1$.
    The Hamiltonian resulting from this modification follows from above.
    We denote the new Hamiltonian constructed from the pre-idled circuit as $\hat{H}$.
    Importantly, the history state is of the form $\ket{\eta} = \ket{\eta_1} + \ket{\eta_2}$, where
    \begin{align*}
        \ket{\eta_1} &\coloneqq \frac{1}{\sqrt{\hat{K}+1}} \sum_{t=0}^{N+T}\ket{\varphi_t}\ket{t},\\
        \ket{\eta_2} &\coloneqq \frac{1}{\sqrt{\hat{K}+1}} \sum_{t=N+T+1}^{\hat{K}}\ket{\varphi_t}\ket{t},
    \end{align*}
    with $\ket{\varphi_t} = \hat{R}_t\ket{\varphi_{t-1}}$, $\ket{\varphi_0} = \ket{x,0^m,+^p}$ and $\ket{t} = |1^t\,0^{\hat{K}-t}\rangle$.

    The non-degenerate ground space of $\hat{H}$ is spanned by $\ket{\eta}$ such that $\hat{H}\ket{\eta} = 0$.
    Additionally, the spectral gap of $\hat{H}$ is $\Omega(1/\hat{K}^3)$~\cite[Lemma 2.2]{gharibian2019complexity}. 

    \subparagraph{Existence of a guiding state.} Now we demonstrate the existence of a semi-classical encoded subset state $\ket{\zeta}$ that has sufficient overlap with the history state $\ket{\eta}$.
    It is clear the history state $\ket{\eta}$ is a subset state~\cite{waite2025complexityb}.
    Furthermore, the component $\ket{\eta_1}$ is a semi-classical encoded subset state by~\cref{lma:partial-history-SCESS}.

    It then follows that the semi-classical encoded subset state $\ket{\zeta}$ can be equivalent to the history state component $\ket{\eta_1}$.
    Moreover, let 
    \begin{equation}
        \ket{\zeta} \coloneqq \frac{1}{\sqrt{Q+1}} \sum_{t=0}^{Q} \ket{\varphi_t}\ket{t}, ~~\text{for}~~ 0 \leq Q \leq N+T.
    \end{equation}
    then,
    \begin{equation}
        \frac{1}{r(n) + q(n) + 1} \leq ~ \abs{\braket{\zeta}{\eta}}^2 = \frac{Q+1}{\hat{K}+1} ~ < 1.
    \end{equation}
    Rearranging slightly and letting $c \coloneqq \hat{K}-Q$ we have
    \begin{equation}
        \frac{Q+1}{\hat{K}+1} = 1 -\frac{c}{p(n)+r(n)} \eqqcolon \delta.
    \end{equation}
    Clearly $\delta \in (0,1]$; then, for sufficiently choices of $r(n)$ and $q(n)$ we have $\delta \in (1/\poly{n}, 1 - 1/\poly{n})$.

    \subparagraph{Perturbed Hamiltonian.} Define $\hat{H}_{\text{out}} = \Pi^{(0)}_1 \otimes \ketbra{1}_{c_{\hat{K}}}$ and note that the smallest eigenvalue of $\hat{H}$ is at least $1/(7\hat{K}^3)$~\cite[Lemma 2.2]{gharibian2019complexity}.
    Apply the Schrieffer-Wolff transformation on $\widetilde{H} = \Delta \hat{H} + \hat{H}_{\text{out}}$ with $\Delta \geq 112 \cdot \hat{K}^3$~\cite{bravyi2016complexity}.
    The result is that $\widetilde{H}$ as a non-degenerate ground space spanned by $\ket{\xi}$.
    From~\cite{bravyi2016complexity} we conclude 
    \begin{equation}\label{eq:history-state-distance-ground-state}
        \norm{\ket{\xi} - \ket{\eta}} = O\left(1/\poly{n}\right),
    \end{equation}
    i.e., for an appropriate choice of $\Delta$ we can bound $\norm{\ket{\xi} - \ket{\eta}}$ from above by an inverse-polynomial $\varepsilon$. 
    
    \subparagraph{State overlap.} To study the overlap between the three states $\ket{\eta}$, $\ket{\zeta}$ and $\ket{\xi}$, we use~\cref{lma:norm-tracking} with $\epsilon_{\xi\eta} = \varepsilon$ and $\delta_{\eta\zeta} = \delta$.
    We have that 
    \begin{align*}
        \abs{\braket{\xi}{\zeta}}^2 &\geq (\sqrt{\delta} - \varepsilon)^2,\\
            &\eqqcolon \kappa.
    \end{align*}
    For sufficiently chosen $\varepsilon$ we have $\kappa \in (1/\poly{n}, 1 - 1/\poly{n})$.
    In typical fashion, the appropriate choices of these parameters stem from $r(n)$, $q(n)$ and $\Delta$.

    \subparagraph{Hamiltonian energy bounds.} It can easily be shown that in each case:
    \begin{align}
        \sc{yes}&: \bra{\eta}\hat{H}_{\text{out}}\ket{\eta} \leq \displaystyle\frac{2^{-f(n)}}{\hat{K}+1}  \label{eq:Yes-case-energy}\\
        \sc{no}&: \bra{\eta}\hat{H}_{\text{out}}\ket{\eta} \geq \displaystyle\frac{1-2^{-f(n)}}{\hat{K}+1} \label{eq:No-case-energy}
    \end{align}
    Standard first-order Schrieffer-Wolf transformation arguments~\cite{bravyi2016complexity} demonstrate that in each case, the lowest eigenvalues of $\widetilde{H}$ are ${2^{-f(n)}}/({\hat{K}+1}) \pm O(1/\Delta)$ and ${(1-2^{-f(n)})}/({\hat{K}+1}) \pm O(1/\Delta)$ respectively.
    By manual renormalisation of the Hamiltonian norm, i.e., from $\|\widetilde{H}\| = O(\poly{n})$ to $\|\widetilde{H}\| \leq 1$ (to satisfy the norm conditions of~\cref{def:GLHP}), that the spectral gap of $\widetilde{H}$ is $\Omega(1/\poly{n})$.
    Combining~\cref{eq:history-state-distance-ground-state} with the uncertainty ranges of~\cref{eq:Yes-case-energy} and~\cref{eq:No-case-energy}, we can also conclude $b' - a' \geq 1/\poly{n}$, where $a'$ and $b'$ are the renormalised energy bounds.
\end{proof}

\begin{remark}
    The parameter $T$ denotes the number of classical gates acting non-trivially only on the first $n+m$ qubits.
    This quantity is determined by the circuit description and may be zero.
    The reduction and guiding-state construction in the proof of \cref{thrm:G6LSH-BPP-hard} do not require $T > 0$, as the pre-idling segment of the history state alone suffices to guarantee the required overlap promise.
\end{remark}

\subsection{Reducing the locality of the problem}
Reducing the locality of the problem can be done in a few ways.
One method is to employ more stringent clock constructions~\cite{kempe2003local, kempe2006complexity}, however, for our purposes we shall consider the method of perturbation gadgets.
A caveat to this method is ensuring the guiding state remains a semi-classical encoded subset state and the overlap with the ground state remains within the desired range.
The intuition behind this idea is outlined in Ref.~\cite{cade2023improved} where it is pointed out that the mediator qubits attached during the perturbative reduction can be identified with a tensor product of isometries from single-qubit states to $O(1)$-qubit states.
Moreover, employing the perturbative approach requires attaching a polynomial number of ancilla qubits to the ground state.
Each attached ancilla qubit is related to one gadget; hence for a set of $w$ gadget ancillae, $\{\ket{\psi_j}\}_{j\in[w]}$, we require 
\begin{equation}\label{eq:gadget-isometry-ancilla}
    \bigotimes_{j\in[w]} \, \ket{\psi_j} \eqqcolon \bigotimes_{j\in[w]} \, V_j\ket{0},
\end{equation}
such that each $\ket{\psi_j}$ is a normalised $O(1)$-qubit state, and $V_j$ is an isometry acting on a single qubit mapping $\ket{0} \mapsto \ket{\psi_j}$.
Analysing the stoquastic gadgets in Ref.~\cite{bravyi2006complexity} reveals that each required $\ket{\psi_j}$ will be of uniform amplitude; thus~\cref{eq:gadget-isometry-ancilla} is a semi-classical encoded subset state.
We denote such gadget ancillae subset states as $\ket{\Psi_{\mathcal{V}}}$.
By~\cref{prop:tensor-product-subset-states} a given guiding state $\ket{\zeta}$ attached to the relevant subset state $\ket{\Psi_{\mathcal{V}}}$ will also be a semi-classical encoded subset state.
The partial consequence of this is that the locality of the stoquastic Hamiltonian is reduced to $2$-local via \refcite{bravyi2006complexity}.
In order to prove~\cref{thrm:G2LSH-BPP-hard} it then suffices to show the overlap between the guiding state and the ground state remains within the desired range.

The perturbative approach can be achieved by successive applications of the Schrieffer-Wolff transformation.
The transformation will result in some Hamiltonian $H_B$ simulating a target Hamiltonian $H_A$ in the sense that their low-energy spectrums are close --- as alluded to in~\cref{sec:preliminaries}.
In addition to having close low-energy spectrums, it has been shown the low-energy states of $H_A$ and $H_B$ are close~\cite{bravyi2016complexity}.
We restate this informally as follows:

\begin{lemma}[\cite{bravyi2016complexity}]
    Let $H$ have a non-degenerate ground state $\ket{\phi_0}$ with a spectral gap $\Delta$.
    Let $\widetilde{H}$ be a $(\eta,\epsilon)$-simulator for $H$ such that $\epsilon < \Delta/2$.
    Then $\widetilde{H}$ has a non-degenerate ground state $|\widetilde{\phi}_0\rangle$ such that 
    \begin{equation}
        \norm{\mathcal{E}\!\ket{\phi_0} - |\widetilde{\phi}_0\rangle} = O(1/\poly{n}),
    \end{equation}
    when $\eta, \epsilon = O(1/\poly{n})$.
\end{lemma}

The inclusion of the encoding on the ground state $\ket{\phi_0}$ is essentially the `attachment of gadget ancillae' qubits discussed above.
A simple application of~\cref{lma:norm-tracking} with $A \coloneqq O(1/\poly{n})$ and $B \coloneqq \delta$ shows 
\begin{equation}
    \abs{\bra{\zeta,\Psi_{\mathcal{V}}}{\widetilde{\phi}_0}\rangle}^2 \in (1/\poly{n}, 1 - 1/\poly{n}).
\end{equation}
Intuitively this makes sense as we expect the distance between the guiding state and the ground state to be small.
The following theorem summarises the results of this section.

\begin{restatable}[]{theorem}{GtwoLSHBPPhard}
    \label{thrm:G2LSH-BPP-hard}
    The \sc{Guided $2$-Local Hamiltonian} problem for stoquastic Hamiltonians is \clw{BPP}{hard} for any $\delta \in (1/\poly{n}, 1 - 1/\poly{n})$.
\end{restatable}

\subsection{Geometrical restrictions}
Similar to reduction in locality, it is possible to construct a sequence of gadgets that enforce geometrical restrictions on the Hamiltonian.
This can be achieved by transforming standard \clsb{BPP}{q} circuits to ones that are \emph{spatially sparse}.
Spatial sparsity loosely refers to the idea that each qubit in the circuit is only acted upon by a small, constant number of gates.
Following the work in Ref.~\cite{waite2025complexitya} and analysing the gadgets, it is clear that similar arguments can be applied to show the \sc{Guided Local Hamiltonian} problem for stoquastic Hamiltonians is \clw{BPP}{hard} when restricted to a square lattice. 

\begin{restatable}[]{theorem}{GsquareLSHBPPhard}
    \label{thrm:GsquareLSH-BPP-hard}
    The \sc{Guided $2$-Local Hamiltonian} problem for stoquastic Hamiltonians on a square lattice is \clw{BPP}{hard} for any $\delta \in (1/\poly{n}, 1 - 1/\poly{n})$.
\end{restatable}

\section{Multi-chotomy Theorem and Completeness}
Cubitt and Montanaro's tetrachotomy theorem~\cite{cubitt2016complexity} captures the complexity classifications of the \sc{Local Hamiltonian} problem for families of interactions generated by a set $\mathcal{M}$.
For a given set of interactions $\mathcal{M}$, the complexity of the corresponding local Hamiltonian family falls into one of four categories: \cl{P}, \clw{NP}{complete}, \clw{StoqMA}{complete}, or \clw{QMA}{complete}.
Kallaugher \emph{et al}.~\cite{kallaugher2024complexity} made a similar classification for the \sc{Product State Local Hamiltonian} problem.
In this the case the result was a dichotomy theorem showcasing that the problem is in \cl{P} if $\mathcal{M}$ is a any $1$-local family, otherwise it is \clw{NP}{complete}.
In an analogous manner, we can extend this idea to the \sc{Guided Local Hamiltonian} problem.
A large class of Hamiltonian families that are typically \clw{QMA}{complete} for the \sc{Local Hamiltonian} problem have been shown to be \clw{BQP}{hard} for the \sc{Guided Local Hamiltonian} problem~\cite{cade2023improved}.
This sets some initial groundwork for a multi-chotomy theorem for the \sc{Guided Local Hamiltonian} problem.
It is trivial to see that the Hamiltonian families typically in \cl{P} will stay in \cl{P} for the \sc{Guided Local Hamiltonian} problem.
The question becomes interesting when considering the \clw{NP}{complete} and \clw{StoqMA}{complete} families.
Since the \clw{NP}{complete} families are classical systems, it remains unclear whether the \cl{BPP} algorithm of Ref.~\cite{gharibian2023dequantizing} suffices to resolve the \sc{Guided Local Hamiltonian} problem variant in a broader sense.
One useful attribute of a subset of the \clw{NP}{complete} Hamiltonian families is that they have computational basis product state ground states, e.g., the Ising model.
We use this fact when discussing our conjecture in a forthcoming subsection.

\subsection{Completeness of local stoquastic Hamiltonians}
We discuss the challenge of proving \emph{completeness} for the \sc{Guided Local Hamiltonian} problem under stoquastic constraints.
While we have established \clw{BPP}{hardness}, completeness requires containment in \cl{BPP}, which is non-trivial due to the problem's reliance on guiding states.
Classical algorithms for ground-state energy estimation typically exploit sample-query access to the guiding state, thereby avoiding quantum state preparation.
However, known results either impose restrictive assumptions or fail to align with the promises specific to our setting.

Ref.~\cite{gharibian2023dequantizing} shows that the \sc{Guided Local Hamiltonian} problem is classically tractable for stoquastic Hamiltonians when both the precision and guiding state overlap are constant.
Their algorithm succeeds with high probability using only efficient sample-query access, without requiring quantum preparation.
Ref.~\cite{zhang2024dequantized} improves upon this by relaxing the norm bound to polynomial in $n$, but introduces two strong constraints: 
\begin{inparaenum}[(i)]
    \item the guiding state must be preparable in constant depth, and
    \item the algorithm requires knowledge of an interval containing the ground-state energy, a promise not provided in the standard problem formulation.
\end{inparaenum}
To address the second constraint, we prove in \cref{app:classical-algorithms} that the cluster expansion methods of Ref.~\cite{mann2024algorithmic} for approximating partition functions of local Hamiltonians can be adapted to decide the \sc{Guided Local Hamiltonian} problem directly, without requiring an energy interval.
We show that when the guiding state overlap, promise gap, and spectral gap are all constant, and the guiding state is preparable in constant-depth with geometrically local gates, there exists a deterministic polynomial-time classical algorithm that decides the problem.
This simplifies and clarifies the approach of Ref.~\cite{zhang2024dequantized}, though it applies only to a restricted special case of \cref{thrm:G6LSH-BPP-hard} (and related general local Hamiltonian results) where these stringent conditions are met.
Note that we still maintain the Hamiltonian norm scaling of $O(\poly{n})$ in this setting, as the cluster expansion method does not require renormalisation to a constant norm.

The algorithm of \citet{stroeks2022spectral} offers another classical approach, capable of learning a constant number of eigenvalues for local stoquastic Hamiltonians with polynomial spectral gaps, assuming access to a state overlapping with a constant number of eigenstates.
While promising, their method relies on two strong assumptions absent in our problem:
\begin{inparaenum}[(i)]
    \item a known lower bound on eigenvalue spacing, and
    \item a bound on the number of eigenstates overlapped by the guiding state.
\end{inparaenum}
Our reduction (from \cref{thrm:G6LSH-BPP-hard}) only guarantees overlap with the ground state.
Moreover, the excited-state structure for Hamiltonians from the Feynman-Kitaev construction remains poorly understood, even more so for perturbed Hamiltonians arising from such constructions.
Hence, the method of Ref.~\cite{stroeks2022spectral} does not clearly apply.

The procedure of \citet{bravyi2015monte} --- a classical analogue of Quantum Phase Estimation --- requires:
\begin{inparaenum}[(a)]
    \item an inverse-polynomial \emph{pointwise} correlation between the guiding state and ground state, and
    \item an efficient classical algorithm for evaluating guiding state amplitudes.
\end{inparaenum}
While our guiding state exhibits at least $1/\poly{n}$ \emph{global} overlap, this does not imply the required pointwise correlation.
Unless the algorithm in \cite{bravyi2015monte} can be adapted to use global overlap alone, it is unsuitable for our setting.
Though, if the given guiding states are indeed pointwise correlated and the problem definition appropriate modified, the algorithm can be applied allowing for a \clw{BPP}{complete} problem.

The Markov Chain technique used to show containment of the \sc{Frustration-Free Local Stoquastic Hamiltonian} problem in \cl{MA} does not suffice here: the ground state may not exhibit low local energy, making local checks ineffective.
Similarly, extensions of Gillespie's Algorithm~\cite{bravyi2023rapidly} also fail, as we lack bounds on or access to eigenstate amplitudes.
For instance, the guiding state $\ket{\zeta} = \sqrt{\delta}\ket{\phi_0} + \sum_{j>0}c_j\ket{\phi_j}$ yields:
\begin{equation*}
    \lambda_0 = \frac{1}{\sqrt{\delta}}\left(\sum_{y} \mel{x}{H}{y}\frac{\braket{y}{\zeta}}{\braket{x}{\phi_0}} - \sum_{j>0}{c_j\lambda_j}\frac{\braket{x}{\phi_j}}{\braket{x}{\phi_0}}\right),
\end{equation*}
but since amplitudes of the eigenstates are unknown, it is unclear how to compute $\lambda_0$, even when $x$ lies in the support of $\ket{\zeta}$.

Given these limitations, \clw{BPP}{completeness} appears unlikely, motivating the search for an intermediate complexity class \cl{X} between \cl{BPP} and \cl{BQP}.
The \emph{Fourier Hierarchy} \cl{FH}~\cite{shi2005quantum}, consisting of classes \cl{FH$_k$} that permit bounded-error quantum circuits with at most $k$ Fourier transforms, offers a potential framework.
Although \cl{FH$_2$} includes Kitaev's Phase Estimation, applying it here is obstructed by the preparation of the guiding state --- especially under stoquastic constraints, where the Perron-Frobenius theorem limits available phases.
Alternatively, one might define a class extending \clsb{BPP}{q} using single-qubit phase states (e.g., $\ket{\phi} = \alpha\ket{0} + \beta e^{i\phi}\ket{1}$), but this sacrifices stoquasticity, since $I - \ketbra{\phi}{\phi}$ is not stoquastic in general.

In the absence of viable classical algorithms, we conjecture that the problem is \cl{X}-complete for some intermediate class \cl{X}.
\Cref{fig:complexity-classes} illustrates its conjectured position in the complexity landscape.
Gate restrictions (e.g., $\{I, X, \Gate{Cnot}, \Gate{Toffoli}\}$) likely prevent a reduced Quantum Phase Estimation procedure from achieving sufficient accuracy.
Nonetheless, advice-based complexity classes,\footnote{Though relevant only if such classes are not as powerful as \cl{BQP}.} or methods such as imaginary-time evolution~\cite{stroeks2022spectral, zhang2024dequantized} combined with cluster expansion~\cite{mann2024algorithmic}, may provide a path forward.
The latter may yield near-optimal bounds in stoquastic cases.

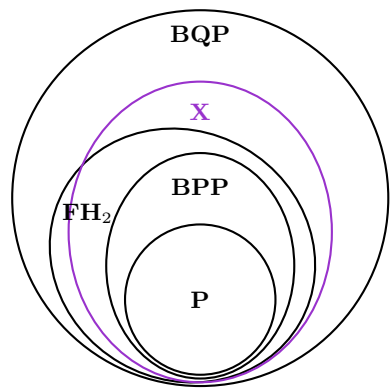
\begin{figure}
    \centering
    \begin{tikzpicture}
        \draw[fill=none, draw=black, thick] (0,1.35) ellipse (2.5 and 2.5);
        \draw[fill=none, draw=black, thick] (0,0.45) ellipse (1.25 and 1.5);
        \draw[fill=none, draw=black, thick] (0,0) circle (1);
        \draw[fill=none, draw=black, thick, rotate around={-30:(-0.2,0.45)}] (-0.3,0.55) ellipse (1.8 and 1.65);
        \draw[fill=none, draw=myViolet, thick] (0,0.9) ellipse (1.75 and 2);
        \node at (0,0) {\cl{P}};
        \node at (0,1.5) {\cl{BPP}};
        \node at (0,3.5) {\cl{BQP}};
        \node at (-1.5,1.15) {\cl{FH$_2$}};
        \node[myViolet] at (0,2.5) {\cl{X}};
    \end{tikzpicture}
    \caption{A diagram of the complexity classes discussed in this work. The conjectured class is denoted by \cl{X} (in violet).}
    \label{fig:complexity-classes}
\end{figure}

\subsection{Duality conjecture}
We conjecture the following trichotomy result for the \sc{Guided Local Hamiltonian} problem.
We are inclined to mention that the table below is not a one-to-one correspondence type argument, but rather lies with similar arguments from Ref.~\cite{cade2023improved}.

\begin{conjecture}
    Let $\mathcal{M}$ be a fixed set of $2$-local interactions.
    There exists a duality between the complexity classification of the local Hamiltonian family generated by \sc{$\mathcal{M}$-Hamiltonian} problem and the \sc{Guided $\mathcal{M}$-Hamiltonian} problem.
    The following table summarises the conjectured duality.
    \begin{table}[!ht]
        \centering
        \setcellgapes{2pt}
        \makegapedcells
        \begin{tabular}{cc}
            \hline
            \makecell{\sc{$\mathcal{M}$-Hamiltonian}} & \makecell{\sc{Guided $\mathcal{M}$-Hamiltonian}} \\\hline\hline
            \makecell{\cl{P}} &\makecell{ \cl{P}} \\
            \makecell{\clw{NP}{complete}} & \makecell{\cl{P}} \\
            \makecell{\clw{StoqMA}{complete}} & \makecell{\clw{X}{complete}}\\
            \makecell{\clw{QMA}{complete}} & \makecell{\clw{BQP}{complete}}\\
            \hline
        \end{tabular}
        \caption{
            A conjectured duality between the complexity classifications of the \sc{$\mathcal{M}$-Hamiltonian} and \sc{Guided $\mathcal{M}$-Hamiltonian} problems.
        }
        \label{tab:duality-conjecture}
    \end{table}
\end{conjecture}

The results of Refs.~\cite{cade2023improved,gharibian2023dequantizing} and the present work support this conjecture.
Note that this conjecture under the confines of the general setting for the \sc{Guided Local Hamiltonian} problem, i.e., then the promise gap is inverse-polynomial.
In the case where efficient classical algorithms (cf. \cref{app:classical-algorithms}) can be utilised, the result collapses to a dichotomy theorem.
Specifically, any problem is in \cl{P} or \cl{BPP}. 

We have partially resolved the conjecture for general stoquastic Hamiltonians, showing \clw{BPP}{hardness} and it is known that the results of Refs.~\cite{cade2023improved} show \clw{BQP}{hardness} for a large class of Hamiltonians, even when restricted to different geometries.
Evidence for the \clw{NP}{complete} instances being in \cl{P} can be provided by the following result when $\mathcal{M}$ is a set of diagonal matrices acting on at most two qubits.

\begin{restatable}{theorem}{GuidedDiagHamiltonian}
    The \sc{Guided $\mathcal{M}$-Hamiltonian} problem is in \cl{P} for any set $\mathcal{M}$ of diagonal matrices acting on at most two qubits.
\end{restatable}

A proof of this result can be found in~\cref{app:proofs}.
The same results holds for any set of interactions $\mathcal{M}$ that can be locally diagonalised by a unitary $U\in {\rm SU}(2)$, if one also amends the definition of semi-classical subset-sates over, globally fixed, but different alphabets.
Then, it is sufficient to present both the Hamiltonian and the guiding state in the same basis, solving the problem in polynomial time.
In the present context, multi-alphabet semi-classical subset states would be appropriate.
Else, a semi-classical encoded subset state with appropriately defined isometries can be used to encode the guiding state.
Upon proving that semi-classical encoded subset states can be efficiently prepared, the \clw{BQP}{hardness} tag can be changed to \clw{BQP}{complete}.

\section{Quantum Hardness}
The \sc{Pinned Local Hamiltonian} problem is a variant of the \sc{Local Hamiltonian} problem where a subset of qubits is fixed in a known state $\ket{\phi}$, called the \emph{pinned subsystem}, while the remaining qubits --- the \emph{free subsystem} --- are optimised over to find the lowest energy state.
Physically, pinning corresponds to fixing a background configuration or static boundary condition.
Therefore, a pinned system has a level of external control, not present in the standard problem.
The idea of the \sc{Pinned} problem is to determine the lowest energy within the pinned subspace, i.e., how computationally hard is to find the ground-state energy of $H_\phi = \expval{H}{\phi}$? 
Using a well-known procedure of embedding general $2$-local Hamiltonians into a larger Hilbert space, \citet{nagaj2021pinned} showed that the \sc{Pinned Stoquastic $3$-Local Hamiltonian} problem is \clw{QMA}{complete}.

Pinning qubits in a stoquastic Hamiltonian can break its structure --- the Hamiltonian within the pinned subspace may no longer be stoquastic, making the free subsystem as hard as the general (\cl{QMA}) case.
The same idea applies in the guided setting when the given Hamiltonian is a pinned stoquastic Hamiltonian.
In both cases, fixing part of the system can amplify complexity.

The \sc{Guided Pinned Local Hamiltonian} problem therefore provides an additional input in the form of a guiding state, which is promised to have a non-negligible overlap with the ground state of the pinned subspace Hamiltonian $H_\phi$.

\begin{restatable}{theorem}{GuidedPinnedLocalHamiltonian}
    \label{thrm:GP2LSH-BQP-hard}
    The \sc{Guided Pinned Stoquastic $3$-Local Hamiltonian} problem is \clw{BQP}{hard} for any $\delta \in (1/\poly{n}, 1 - 1/\poly{n})$.
\end{restatable}

\begin{proof}
    Consider an instance of the \sc{Guided $2$-Local Hamiltonian} problem with Hamiltonian $H$, over $n$ qubits, generated by the set of interactions $\{X,Z,XX,ZZ\}$ such that $H$ has a ground state $\ket{\varphi_0}$ and (SCESS) guiding state $\ket{\zeta}$ with overlap $\delta$~\cite{cade2023improved}.
    Split $H$ into stoquastic ($V$) and non-stoquastic ($P$) parts, 
    \begin{multline*}
        H = \underbrace{\sum_{u,v} J_{uv}^{\leq 0} X_uX_v + L_{uv}Z_uZ_v + \sum_{u} f_{u}^{\leq 0}X_u + h_{u}Z_u}_{\eqqcolon V} \\[0.2cm] + \underbrace{\sum_{u,v} J_{uv}^{> 0} X_uX_v + \sum_{u} f_{u}^{> 0}X_u}_{\eqqcolon P}.
    \end{multline*}
    Embed $H$ in a larger Hilbert space by introducing an auxiliary qubit $q$, fixed in the state $\ket{-}$, and define the Hamiltonian
    \begin{equation*}
        \hat{H} = V\otimes I - P\otimes X_q.
    \end{equation*}
    The auxiliary qubit $q$ is coupled to the non-stoquastic part of the Hamiltonian and represents the pinned subsystem.
    The ground state of $\hat{H}$ is given by the tensor product of the ground state of $H$ and the auxiliary qubit, i.e., $\ket{\varphi_0}\ket{-}$.
    It then follows that the guiding state $\ket{\zeta}\ket{-}$ has the same overlap with the ground state of $\hat{H}$, as $\ket{\zeta}$ has with $\ket{\varphi_0}$. 
    The new guiding state is also a semi-classical encoded subset state.
    Therefore, if $\lambda_0$ is at most $a$ (the lower threshold of the original \clw{BQP}{hard} instance), then the ground-state energy of $\hat{H}$ is at most $a$.
    Conversely, if $\lambda_0$ is at least $b$ (the upper threshold of the original \clw{BQP}{hard} instance), then the ground-state energy of $\hat{H}$ is at least $b$.
\end{proof}

This Theorem shows the \textsc{Guided Pinned Stoquastic Local Hamiltonian} problem is not contained in \cl{BPP} unless $\cl{BPP} = \cl{BQP}$.
Since general local Hamiltonians contain local stoquastic Hamiltonians, the \sc{Guided Pinned Local Hamiltonian} problem is also \clw{BQP}{hard}.
Applications of perturbative gadget reductions fail for the pinned Hamiltonians~\cite{nagaj2021pinned}, and the additional of the auxiliary qubit destroys the underlying geometry of the original Hamiltonian in the sense that if $H$ is geometrically local, then $\hat{H}$ is not.

\section{Conclusion}
In this work, we extended the \sc{Guided Local Hamiltonian} problem~\cite{cade2023improved} to local stoquastic Hamiltonians.
The variant studied here differs from that originally conceived by Bravyi~\cite{bravyi2015monte} in that the guiding state is provided as part of the input.
It is known that when a verifier receives a local Hamiltonian and, in addition, a description of a state that is promised to overlap with the ground state, the problem of determining the ground-state energy becomes \clw{BQP}{hard}.
This result extends to a range of Hamiltonian families, including those relevant to condensed-matter systems, and Hamiltonians are geometrically constrained.
We have shown that when restricted to stoquastic Hamiltonians, the problem becomes \clw{BPP}{hard}, even when the Hamiltonian is $2$-local and confined to a square lattice.

\begin{result*}[\cref{thrm:G2LSH-BPP-hard}]
    The \sc{Guided $2$-Local Hamiltonian} problem for stoquastic Hamiltonians is \clw{BPP}{hard} for any $\delta \in (1/\poly{n}, 1 - 1/\poly{n})$, even when restricted to a square lattice.
\end{result*}

Our results naturally complement the previous \clw{BQP}{hardness} result for the \sc{Guided Local Hamiltonian} problem.
The existence of a semi-classical subset state with non-negligible overlap with the ground state of a local stoquastic Hamiltonian follows from applying the Feynman-Kitaev construction to \clsb{BPP}{q} circuits.
In particular, the history state defined for part of the circuit's evolution is a subset state.
Since only classically reversible gates are used, the history state admits a succinct description.
The guiding state remains a semi-classical encoded subset state even after applying perturbative gadget reductions that reduce the Hamiltonian's locality to two.
We find that the overlap between the guiding and ground states remains within the desired range under these reductions.
It remains unclear whether we can conclusively say there is no quantum advantage for the \sc{Guided Local Hamiltonian} problem, with inverse-polynomial precision, when restricted to stoquastic Hamiltonians.
Moreover, it is uncertain whether the semi-classical subset states can fully trivialise the problem in the stoquastic regime.
Our reduction shows that, in the presence of a suitable guiding state, the problem of estimating the ground-state energy of a stoquastic Hamiltonian is universal for classical probabilistic computation.
In particular, any efficient probabilistic algorithm can be embedded into such an instance, and hence no deterministic classical algorithm is expected to solve the problem unless $\cl{BPP} = \cl{P}$.

In the mild setting where the promise gap and overlap are constant, the problem is known to be solvable via efficient classical algorithms~\cite{gharibian2023dequantizing}.
We have proven a new but simple result demonstrating that within this mild parameter regime, when the guiding state is preparable in constant depth, there exists a deterministic polynomial-time classical algorithm for the problem~\cite{mann2024algorithmic} (see \cref{app:classical-algorithms}).
This result applies not just to local stoquastic Hamiltonians.
Interestingly, we do not require the strong norm bound condition of Ref.~\cite{gharibian2023dequantizing}; the idea of relaxing the norm bound to $O(\poly{n})$ is inspired by Ref.~\cite{zhang2024dequantized}, though in contrast to this work we do not require the additional promise of an energy interval containing the ground-state energy.

To summarise, we proposed a trichotomy conjecture concerning the classification of Hamiltonian families under the \sc{Guided Local Hamiltonian} framework.
Specifically, we conjecture that the \sc{Guided Local Hamiltonian} problem for stoquastic Hamiltonians is \clw{X}{complete} for some complexity class \cl{X} that lies strictly between \cl{BPP} and \cl{BQP}.
While identifying the exact class \cl{X} may be difficult, establishing upper and lower bounds could be more feasible.

In addition to our \clw{BPP}{hardness} result, we also showed that the \sc{Guided Pinned $3$-Local Hamiltonian} problem is \clw{BQP}{hard} for stoquastic Hamiltonians.
This version assumes that the Hamiltonian has a subset of qubits fixed in a known state, termed the \emph{pinned subsystem}, while the remaining qubits remain free.
Interestingly, the Hamiltonian restricted to the pinned subspace need not remain stoquastic, rendering the free subsystem as hard as general local Hamiltonians.
Unfortunately, we are unaware of any techniques to further reduce the locality or geometry of the Hamiltonian in this case without compromising the pinned feature.

\subparagraph{Discussion.} 
One open problem is determining the position of \cl{SStoqMA} (subset state \cl{StoqMA}) within the complexity hierarchy.
With our current understanding of \cl{StoqMA}, it remains unknown whether it is equivalent to\cl{SStoqMA}, in contrast to the equivalence of \cl{SQMA} with \cl{QMA}~\cite{grilo2015qma}.
The difficulty here also arises from a weak grasp on error reduction techniques for \cl{StoqMA}.
It is straightforward, following the proof of Ref.~\cite[Theorem 12]{grilo2015qma}, to show that \cl{SStoqMA} $\subseteq$ \cl{StoqMA}; however, the reverse direction presents challenges in the completeness portion.
In particular, by~\cite[Lemma 8]{grilo2015qma}, it follows that 
\begin{equation*}
    \expval{V^\dagger \Pi_+ V}{S} \geq \alpha(n) - \left(1 - \frac{1}{128 (p(n) +3)}\right).
\end{equation*}
Here, $V$ is the verifier's circuit and $\Pi_+$ is the measurement operator.
This bound implies that the acceptance probability of the subset state $\ket{S}$ may not be guaranteed to be $\Omega(1/\poly{n})$.
To elaborate, if $\alpha(n) = 1 - \epsilon(n)$ with $\epsilon(n) \in [0,1/2)$, then
\begin{equation*}
    \expval{V^\dagger \Pi_+ V}{S} \gtrsim \frac{1}{p(n)} - \epsilon(n);
\end{equation*}
the circuit will accept $\ket{S}$ with probability $\Omega(1/\poly{n})$ only if $\epsilon(n) \leq 1/(2p(n))$.
So, is \cl{SStoqMA} strictly weaker than \cl{StoqMA}?

Another problem, already hinted at in Ref.~\cite{grilo2015qma}, is determining the complexity of computing the energy of subset states.
Inspired by recent work on product states and local Hamiltonians~\cite{kallaugher2024complexity}, consider the problem of deciding, given a local Hamiltonian, whether there exists a subset state such that $\mel{S}{H}{S} \leq a$, or whether for all subset states $\mel{S}{H}{S} \geq b$, promised one of the cases holds.
Interestingly, this was recently resolved for a subclass known as \emph{weight-$k$} states by Bremner \emph{et al}.~\cite{bremner2025parameterized}.
They showed that the \sc{Weight-$k$ $l$-Local Hamiltonian} problem belongs to \cl{QW[1]}, a parameterised quantum complexity class~\cite{bremner2022quantum} analogous to \cl{W[1]}~\cite{downey2013fundamentals}.
The same question can be asked for semi-classical subset states.
Ref.~\cite[Lemma 8]{grilo2015qma} suggests the problem may even be \clw{QMA}{hard}.
Though, a parameterised perspective may offer a more fruitful approach.

An additional open problem is to determine how to efficiently construct guiding states for candidate Hamiltonians. 
Since constructing a low-energy product state is already \clw{FNP}{hard} for $2$-local Hamiltonian families~\cite{kallaugher2024complexity}, efficient exact methods are unlikely to exist in general.
Though tractable exceptions are known, for example, commuting frustration-free Hamiltonians admit polynomial-time ground-state preparation algorithms, placing this problem in \cl{FBQP}~\cite{schwarz2013information, massar2021total}.
Given the inherently classical nature of stoquastic Hamiltonians, it is plausible that effective heuristics can be developed for this family.
This could be particularly valuable for certain fermionic systems where the Hamiltonian is stoquastic.

Finally, we conjecture that the \sc{Guided Local Hamiltonian} problem for the transverse-field Ising model is \clw{BPP}{hard}. 
We anticipate that the reduction techniques of Bravyi and Hastings~\cite{bravyi2016complexity}, which encode general $2$-local stoquastic Hamiltonians into transverse-field Ising models on degree-$3$ graphs via Schrieffer-Wolff perturbative gadgets, can be adapted to our setting, with the guiding state remaining a semi-classical subset state with non-negligible overlap throughout. 
Though this reduction utilises global encodings, which may complicate the construction of the guiding state description in polynomial time.
An immediate corollary of this conjecture however, is the \clw{BPP}{hardness} for the families of local Hamiltonians that are typically \clw{StoqMA}{complete} in the Cubitt-Montanaro tetrachotomy~\cite{cubitt2016complexity}.
We leave a formal proof as an open problem.

\textbf{Data availability.} 
No data was generated or analysed in this study.

\textbf{Acknowledgements.}
GW would like to thank Yuval Sanders, Ryan Mann, Michael Bremner and Christopher Howell for helpful discussions and feedback.
GW was supported by a scholarship from the Sydney Quantum Academy and also supported by the ARC Centre of Excellence for Quantum Computation and Communication Technology (CQC2T), project number CE170100012.

\bibliographystyle{apsrev4-1}
\bibliography{ref}

\onecolumngrid
\appendix
\section{Deferred Proofs}\label{app:proofs}

\BPPqBPP*

\begin{proof}
    Informally, this result follows from Ref.~\cite{bravyi2006merlin} and discussions in Ref.~\cite{waite2025complexitya}.
    This may not be immediately clear, thus for clarity, we provide a formal proof.
    To prove equivalency we need to show that \clsb{BPP}{q} $\subseteq$ \cl{BPP} and \cl{BPP} $\subseteq$ \clsb{BPP}{q}.\newline

    \paragraph{\clsb{BPP}{q} $\subseteq$ \cl{BPP}.}
    Let $L = (L_{\sc{yes}}, L_{\sc{no}})$ be a promise problem in \clsb{BPP}{q}.
    This means that there exists a classically reversible quantum circuit $M_{|x|}$ that acts on input $\ket{x}$ along with classical ancillae $\ket{0^m}$ and quantum ancillae $\ket{+^p}$.
    The role of the $\ket{+^p}$ ancillae is to provide randomness via a uniform superposition.

    Since the quantum gates involved (e.g., $X$, $\Gate{Cnot}$, and $\Gate{Toffoli}$) are classically reversible, and the only source of quantum behaviour comes from the superposition state $\ket{+^p}$, we can replace the quantum randomness with classical randomness.
    This shows that any computation performed by the quantum circuit $M_{|x|}$ can be simulated by a classical probabilistic circuit.
        
    \subparagraph{Case 1:}
    By definition, if $x \in L_{\sc{yes}}$, the quantum circuit $M_{|x|}$ outputs $1$ with probability at least $\frac{2}{3}$.
    As the quantum circuit can be simulated by a classical probabilistic circuit with random coin flips instead of $\ket{+^p}$, the same success probability can be achieved classically.
    Therefore, there exists a classical probabilistic circuit that outputs $1$ with probability at least $\frac{2}{3}$ for $x \in L_{\sc{yes}}$.

    \subparagraph{Case 2:}
    Similarly, if $x \in L_{\sc{no}}$, the quantum circuit $M_{|x|}$ outputs $1$ with probability at most $\frac{1}{3}$.
    Again, since the quantum computation can be simulated classically by replacing the quantum randomness with classical randomness, the classical probabilistic circuit will also have an output probability of at most $\frac{1}{3}$ for $x \in L_{\sc{no}}$. \newline

    \paragraph{\cl{BPP} $\subseteq$ \clsb{BPP}{q}.}
    Now, let $L = (L_{\sc{yes}}, L_{\sc{no}})$ be a promise problem in \cl{BPP}.
    By definition, there exists a classical probabilistic circuit that solves this problem with the required success probability bounds.

    We will now simulate the classical probabilistic circuit using a quantum circuit.
    The classical random bits used by the probabilistic machine can be replaced by initializing the quantum circuit with ancillae in the state $\ket{+^p}$, which provides the required randomness through a uniform superposition.
    The classical gates in the probabilistic circuit can be mapped directly to the corresponding reversible quantum gates (e.g., Toffoli).

    \subparagraph{Case 1:}
    If $x \in L_{\sc{yes}}$, the classical probabilistic machine outputs $1$ with probability at least $\frac{2}{3}$.
    In the quantum circuit, the randomness is simulated by the $\ket{+^p}$ ancillae, and the classical gates are reversible.
    Therefore, the quantum circuit outputs $1$ with probability at least $\frac{2}{3}$ for $x \in L_{\sc{yes}}$.

    \subparagraph{Case 2:}
    If $x \in L_{\sc{no}}$, the classical probabilistic machine outputs $1$ with probability at most $\frac{1}{3}$.
    The quantum circuit, which simulates the classical circuit, also outputs $1$ with probability at most $\frac{1}{3}$, as the quantum ancillae provide the same type of randomness and the gates are reversible.
\end{proof}

\normtracking*

\begin{proof}
    We recall some basic facts and definitions:
    \begin{inparaenum}[(i)]
        \item $\ket{u} = \ket{u} - \ket{v} + \ket{v}$;
        \item $\abs{\braket{u}{v}} \leq \norm{\ket{u}} \norm{\ket{v}}$;
        \item ${\abs{x+y} \geq \abs{\abs{x} - \abs{y}}}$.
    \end{inparaenum}
    We start with the bound from below.
    \begin{equation*}
        \abs{\braket{a}{c}} = \big|{\braket{b}{c} + (\bra{a}-\bra{b})\ket{c}}\big| \geq \big|{\abs{\braket{b}{c}} - \abs{(\bra{a}-\bra{b})\ket{c}}}\big| \geq \abs{\sqrt{\delta_{bc}} - \norm{\ket{a}-\ket{b}}} \geq \abs{\sqrt{\delta_{bc}} - \epsilon_{ab}},
    \end{equation*}
    hence, provided $\epsilon_{ab} \leq \sqrt{\delta_{bc}}$, it follows that 
    \begin{equation*}
        \abs{\braket{a}{c}}^2 \geq (\sqrt{\delta_{bc}} - \epsilon_{ab})^2.
    \end{equation*}
    In the event that $\epsilon_{ab} > \sqrt{\delta_{bc}}$, we have $\abs{\braket{a}{c}} \geq 0$.
    The upper bound is obtained by applying the triangle inequality:
    \begin{equation*}
        \abs{\braket{a}{c}} = \abs{\braket{b}{c} + (\bra{a}-\bra{b})\ket{c}} \leq \abs{\braket{b}{c}} + \abs{(\bra{a}-\bra{b})\ket{c}} \leq \abs{\braket{b}{c}} + \norm{\ket{a}-\ket{b}} \leq \abs{\sqrt{\delta_{bc}} + \epsilon_{ab}}.
    \end{equation*}
\end{proof}

\tensorproductsubsetstates*

\begin{proof}
    The proof of this proposition follows trivially from considering the tensor product of two semi-classical subset states.
    Let $\ket{S_\mathcal{V}}$ and $\ket{T_\mathcal{U}}$ be two semi-classical subset states over $(S, \mathcal{V})$ and $(T, \mathcal{U})$ respectively.
    We have that $S \subset \B^n$, $T \subset \B^m$ where $\abs{S} = \poly{n}$ and $\abs{T} = \poly{m}$.
    Additionally, $\mathcal{V} = \{V_j\}_{j\in[n]}$ and $\mathcal{U} = \{U_k\}_{k\in[m]}$ are sets of isometries from single qubit states to $O(1)$-qubit states.
    The tensor product of these states is
    \begin{align*}
        \ket{S_\mathcal{V}} \ket{T_\mathcal{U}} &= \frac{1}{\sqrt{\abs{S}\abs{T}}} \left(\sum_{x \in S} \bigotimes_{j\in[n]} V_j\ket{x}\right)\left( \sum_{y \in T} \bigotimes_{k\in[m]} U_k\ket{y}\right),\\
            &= \frac{1}{\sqrt{\abs{S}\abs{T}}} \sum_{x \in S} \sum_{y \in T} \left( \bigotimes_{j\in[n]} V_j\ket{x}\right) \left( \bigotimes_{k\in[m]} U_k\ket{y}\right),\\
            &= \frac{1}{\sqrt{\abs{S}\abs{T}}} \sum_{z \in S \times T} \bigotimes_{l\in[n+m]} W_l \ket{z_l},\\
            &= \frac{1}{\sqrt{\abs{R}}} \sum_{z \in R} \bigotimes_{l\in[p]} W_l \ket{z_l},\\
    \end{align*}
    where $R = S \times T \subset \B^{p}$, with $p = n+m$ and $\abs{R} = \abs{S}\abs{T} = \poly{n}\poly{m} \leq \poly{p}$.
    We also have a set of $p$ isometries $\mathcal{W} = \{W_l\}_{l\in[p]} = \mathcal{V}\cup\mathcal{U}$.
    Clearly, the tensor product of two semi-classical encoded subset states is itself a semi-classical encoded subset state.
    This holds for a polynomial number of such tensor products. 
\end{proof}

\classicallyreversiblegateonsubsetstate*

\begin{proof}
    This follows trivial from the fact that each classical gate is a bijective map.
    Let $R$ be a classical reversible gate acting on a subset state $\ket{S}$.
    Then we have 
    \begin{equation*}
        R\ket{S} = \frac{1}{\sqrt{\abs{S}}} \sum_{x \in S} R\ket{x} = \frac{1}{\sqrt{\abs{S}}} \sum_{y \in R(S)} \ket{y} = \ket{R(S)}.
    \end{equation*}
\end{proof}

\historystatesubsetstate*

\begin{proof}
    The proof follows from the fact that for each $k \in [K]$, the sets $(R_k \cdots R_1(S))\times \{1^k\,0^{K-k}\}$ are mutually orthogonal.
    Let $S_k = (R_k \cdots R_1(S))\times \{1^k\,0^{K-k}\}$ and $\mathcal{S} = \bigcup_k S_k$, then 
    \begin{equation*}
        \frac{1}{\sqrt{K}} \sum_{k \in [K]} \ket{S_k} = \frac{1}{\sqrt{K}} \sum_{k \in [K]} \sum_{y \in S_k} \frac{1}{\sqrt{\abs{S}}}\ket{y} = \frac{1}{\sqrt{\abs{S}\cdot K}} \sum_{x \in \mathcal{S}} \ket{x} \eqqcolon \ket{\mathcal{S}}.
    \end{equation*}
    For $k \in [K]$, an element of $S_k$ is a binary string of the form $y = R_k \cdots R_1(x) \mathbin\Vert \{\underbrace{11\dots 1}_{k}\,\underbrace{00\dots 0}_{K-k}\}$ where $x \in S$.
\end{proof}

\partialhistorySCESS*

\begin{proof}
    The proof follows from the proof of~\cref{lma:history-state-subset-state} in the same essence as~\cref{cor:history-state-semi-classical-subset-state}.
    Note that the classically reversible gates commute with the isometries since the non-trivial supports of each are disjoint.
\end{proof}

The proof of the following theorem makes use of the following lemma from Ref.~\cite{cubitt2016complexity},

\begin{lemma}
    Let $\mathcal{M}$ be a set of diagonal matrices acting on at most two qubits.
    If for any $M \in \mathcal{M}$ is $1$-local, then the \sc{$\mathcal{M}$-Hamiltonian} problem is in \cl{P}, else the \sc{$\mathcal{M}$-Hamiltonian} problem is in \clw{NP}{complete}.
\end{lemma}

It is well-known that the ground state of classical local Hamiltonians is a product state --- a computational basis state.
We denote the bit string corresponding to the ground state as $\g$ and the corresponding product state as $\ket{\phi_0}$.

\GuidedDiagHamiltonian*

\begin{proof}
    There exists a semi-classical subset state $\ket{S}$ such that 
    \begin{equation*}
        \abs{\braket{S}{\phi_0}}^2 \coloneqq \frac{\ind[\phi_0 \in S]}{\abs{S}},
    \end{equation*}
    where $\ind[\phi_0 \in S]$ is $1$ if $\g\in S$ and $0$ otherwise.
    Clearly $\abs{\braket{S}{\phi_0}}^2 \in [0,1]$.
    We can rule out the two edge cases, namely of $\abs{\braket{S}{\phi_0}}^2 = 0$ the problem remains \clw{NP}{complete} and if $\abs{\braket{S}{\phi_0}}^2 = 1$ the problem is trivially in \cl{P}.
    We therefore restrict our attention to the case $\abs{\braket{S}{\phi_0}}^2 \geq 1/\poly{n}$.
    The crucial point to realise that when the fidelity is larger than zero, it is guaranteed by orthogonality that $\g\in S$.
    Thus, to find the ground state of the Hamiltonian it suffices to solve the minimisation task:
    \begin{equation*}
        \min_{x\in S} \mel{x}{H}{x} = \min_{x\in S} \sum_j \mel{x}{h_j}{x},
    \end{equation*}
    where $h_j$ is a diagonal matrix.
    This can be done classically in polynomial time.
\end{proof}

\section{Stoquastic Arthur Definition}\label{app:StoqP}
Here we present a formal definition of the class \cl{StoqP}, relating to the underlying verification class for Arthur in \cl{StoqMA}.
We referred to the naive definition of \cl{StoqP} as simply \cl{StoqMA} without a proof --- this is sometimes how \cl{QMA} and \cl{BQP} can be related by definition.

To formally define \cl{StoqP}, we require an adaptation of classically reversible quantum circuits.
Specifically, we must modify the output measurement to be in the $X$-basis, rather than the computational basis.
Therefore, let $V_n^{(+)}$ denote the circuits as in~\cref{def:CRQVC} with the output measurement in the $X$-basis, otherwise the definition remains the same.
We call this set of circuits a \emph{classically reversible quantum circuit with $X$-basis output measurement}, or CRQVC\textsuperscript{(+)} for short.

\begin{definition}[\cl{StoqP}]
    Let $L = (L_{\sc{yes}}, L_{\sc{no}})$ be a promise problem.
    A problem $L$ belongs to the class \cl{StoqP} if there exists a polynomial-time generated CRQVC\textsuperscript{(+)} family $\mathcal{V}^{(+)} = \{V_n^{(+)} : n \in \mathbb{N}\}$, where each circuit $V_n^{(+)}$ acts on $n + m + p$ input bits and produces one output bit, such that:
    \begin{itemize}
        \item If $x \in L_{\sc{yes}}$, then $\Pr[V_{|x|}^{(+)}(x) = {\tt +}] \geq \alpha(|x|)$,
        \item If $x \in L_{\sc{no}}$, then $\Pr[V_{|x|}^{(+)}(x) = {\tt +}] \leq \beta(|x|)$,
    \end{itemize}
    where $\frac{1}{2} \leq \beta(|x|) < \alpha(|x|) \leq 1$, satisfying $\alpha(|x|) - \beta(|x|) = \Omega(1/\poly{|x|})$.
\end{definition}

\subsection{Justification of the Definition of StoqP}
To justify this choice of complexity class, consider the following definitions and theorem, inspired by Ref.~\cite{marriott2003non}.

\begin{definition}
    Let $\hat{\exists}$ be the quantum existential quantifier, and $\cl{K}$ denote a complexity class.
    A promise problem $L = (L_{\sc{yes}}, L_{\sc{no}})$ belongs to $\hat{\exists}\cdot\cl{K}$ if and only if:
    \begin{itemize}
        \item If $x \in L_{\sc{yes}} ~ \implies ~ \exists \ket{\psi} \in (\mathbb{C}^2)^{\otimes w(|x|)} ~ : (x, \ket{\psi}) \in J_{\sc{yes}}$,
        \item If $x \in L_{\sc{no}} ~ \implies ~ \forall \ket{\psi} \in (\mathbb{C}^2)^{\otimes w(|x|)} ~ : (x, \ket{\psi}) \in J_{\sc{no}}$,
    \end{itemize}
    for some promise problem $J = (J_{\sc{yes}}, J_{\sc{no}}) \in \cl{K}$ and polynomially bounded function $w$.
\end{definition}

\begin{definition}[\cl{StoqP}-consistency]
Let $J = (J_{\sc{yes}}, J_{\sc{no}})$ be a promise problem over inputs of the form $(x,\ket{\psi})$, where $x \in \B^*$ and $\ket{\psi}$ is a quantum state on $w = \poly{|x|}$ qubits.
A uniform family of stoquastic circuits $\mathcal{V}^{(+)} = \{V_n^{(+)}\}$ is said to be \emph{\cl{StoqP}-consistent with $J$} if for all $(x,\ket{\psi})$ in the promise of $J$ with $|x| = n$:
\begin{itemize}
    \item If $(x,\ket{\psi}) \in J_{\sc{yes}}$, then $\Pr[V_n^{(+)}(x,\ket{\psi}) = \mathtt{+}] \geq \alpha(n)$,
    \item If $(x,\ket{\psi}) \in J_{\sc{no}}$, then $\Pr[V_n^{(+)}(x,\ket{\psi}) = \mathtt{+}] \leq \beta(n)$,
\end{itemize}
where $1/2 \le \beta(n) < \alpha(n) \le 1$ and $\alpha(n) - \beta(n) \ge 1/\poly{n}$.
\end{definition}

\begin{theorem}\label{thm:stoqma-equals-stoqp}
    \begin{equation*}
        \hat{\exists}\cdot\cl{StoqP} = \cl{StoqMA}.
    \end{equation*}
\end{theorem}

\begin{proof}
    We prove both inclusions separately.

    \medskip
    \noindent\textbf{($\hat{\exists}\cdot\cl{StoqP} \subseteq \cl{StoqMA}$).}
    Let $L = (L_{\sc{yes}}, L_{\sc{no}}) \in \hat{\exists}\cdot\cl{StoqP}$. 
    Then there exists a promise problem $J = (J_{\sc{yes}}, J_{\sc{no}})$ over inputs $(x,\ket{\psi})$ and a uniform family of stoquastic circuits $\{V_n^{(+)}\}$ that is \cl{StoqP}-consistent with $J$, such that:
    \begin{align*}
        x \in L_{\sc{yes}} &\implies \exists \ket{\psi} : (x,\ket{\psi}) \in J_{\sc{yes}} \implies \exists \ket{\psi} : \Pr[V_n^{(+)}(x,\ket{\psi}) = \mathtt{+}] \geq \alpha(n), \\
        x \in L_{\sc{no}} &\implies \forall \ket{\psi} : (x,\ket{\psi}) \in J_{\sc{no}} \implies \forall \ket{\psi} : \Pr[V_n^{(+)}(x,\ket{\psi}) = \mathtt{+}] \leq \beta(n).
    \end{align*}
    A $\cl{StoqMA}$ verifier simulates $V_n^{(+)}$ on input $(x,\ket{\psi})$, where $\ket{\psi}$ is provided by Merlin. 
    That is, on input $x$, Merlin sends a witness $\ket{\psi}$ and Arthur runs $V_n^{(+)}$ on $(x,\ket{\psi})$ and answers accordingly.
    If $x \in L_{\sc{yes}}$, then there exists a witness $\ket{\psi}$ that causes Arthur to accept with probability at least $\alpha(n)$, and if $x \in L_{\sc{no}}$, then for all witnesses $\ket{\psi}$, Arthur accepts with probability at most $\beta(n)$.
    Therefore $L \in \cl{StoqMA}$.

    \medskip
    \noindent\textbf{($\cl{StoqMA} \subseteq \hat{\exists}\cdot\cl{StoqP}$).}
    Let $L = (L_{\sc{yes}}, L_{\sc{no}}) \in \cl{StoqMA}$. 
    Then there exists a stoquastic verification circuit $S$ (from the family $\{S_n\}$) for $L$.
    Let $S_{\text{Arthur}}$ be Arthur's component of the verification circuit $S$ that acts on the input $x$ and the witness $\ket{\psi}$, and produces an output bit.
    It follows that for every instance $x$:
    \begin{align*}
        x \in L_{\sc{yes}} &\implies \exists \ket{\psi} \text{ with } \Pr[S_{\text{Arthur}}(x,\ket{\psi}) = \mathtt{+}] \ge \alpha(n), \\
        x \in L_{\sc{no}} &\implies \forall \ket{\psi},\ \Pr[S_{\text{Arthur}}(x,\ket{\psi}) = \mathtt{+}] \le \beta(n).
    \end{align*}
    The circuit is \cl{StoqP}-consistent with a promise problem $J = (J_{\sc{yes}}, J_{\sc{no}})$.
    Hence, we can see that $J \in \cl{StoqP}$ and that $L \in \hat{\exists}\cdot\cl{StoqP}$.
\end{proof}

\subsection{Relation to BPP\textsubscript{q}}
Then $\cl{StoqP} = \cl{StoqP}(\alpha, \beta)$.
Note that $\alpha$ is necessarily bounded away from $\frac{1}{2}$ however $\beta$ is not.
Our arguments prove that 
\begin{equation*}
    \cl{StoqP}(\alpha,\beta) \subseteq \clsb{BPP}{q}(2\alpha - 1, 2\beta - 1).
\end{equation*}
To show the containment it suffices to prove that a \clsb{BPP}{q} machine can simulate a \cl{StoqP} machine for a given instance.

Fix $x \in \B^n$, let $V = V_{|x|}^{(+)}$ and $\ket{\phi_x} = \ket{x, 0^m, +^p}$ for brevity. 
Since every gate in $\{I, X,\Gate{Cnot},\Gate{Toffoli}\}$ is self-inverse and classically reversible, the circuit $V^\dagger X_1 V$ is again a polynomial-size circuit over the same gate set. 
Consider the operator $W \coloneqq V^\dagger X_1 V$. 
From the fact $\ketbra{+}{+}_1 = \frac{1}{2}(I + X_1)$, we observe that 
\begin{equation*}
    \mel{\phi_x}{W}{\phi_x} = 2\Pr[V_{|x|}^{(+)}(x) = {\tt +}] - 1.
\end{equation*}
It therefore suffices to realise the quantity $\mel{\phi_x}{W}{\phi_x}$ as the acceptance probability of a CRQVC (cref{def:CRQVC}).

Since $W$ is a circuit over classically reversible gates, it acts as a permutation on the computational basis. Write $a \coloneqq (x,0^m) \in \{0,1\}^{n+m}$. 
Expanding the $\ket{+}$ ancillae in the computational basis gives $\ket{\phi_x} = 2^{-p/2}\sum_{z\in\{0,1\}^p}\ket{a,z}$. 
Hence
\begin{equation*}
    \mel{\phi_x}{W}{\phi_x} = 2^{-p}\sum_{z,z' \in \{0,1\}^p} \mel{a,z'}{W}{a,z}.
\end{equation*}
Because $W$ permutes computational basis states, for each $z$ there is a unique basis string $f(a,z)$ such that $W\ket{a,z}=\ket{f(a,z)}$. 
Writing $f(a,z)=(y_z,w_z)$ with $y_z \in \{0,1\}^{n+m}$ and $w_z \in \{0,1\}^p$, the matrix element $\mel{a,z'}{W}{a,z}$ is equal to $1$ if and only if $(y_z,w_z)=(a,z')$, and is otherwise $0$. 
Therefore
\begin{equation*}
    \mel{\phi_x}{W}{\phi_x} = 2^{-p}\bigl|\{z \in \{0,1\}^p : y_z = a\}\bigr|.
\end{equation*}
This is exactly the probability that, after applying $W$ to the state $\ket{a}\ket{+^p}$ and measuring in the computational basis, the first $n+m$ output qubits are equal to $a=(x,0^m)$.

We now construct a \clsb{BPP}{q} machine that simulates this computation.
On input $x$, the circuit $M_{|x|}$ first appends $n$ additional ancilla qubits initialised to $\ket{0^n}$ and copies the classical input string $x$ into them using \Gate{Cnot} gates from the original input register. 
This is valid because the input is a computational basis string. 
The circuit then applies $W$ to the original $n+m+p$ qubits. 
After this, it reversibly compares the first $n$ qubits of the transformed register with the stored copy of $x$, and also checks that the next $m$ qubits are all equal to $0$. 
Using standard reversible classical computation over the gate set $\{I, X,\Gate{Cnot},\Gate{Toffoli}\}$, these checks may be combined into a single output qubit which is set to ${\tt 1}$ exactly when the first $n+m$ qubits equal $x,0^m$. 
Thus $M_{|x|}$ is a valid CRQVC, and by the argument above its acceptance probability is precisely
\begin{equation*}
    \Pr[M_{|x|}(x) = {\tt 1}] = \mel{\phi_x}{W}{\phi_x} = 2\Pr[V_{|x|}^{(+)}(x) = {\tt +}] - 1.
\end{equation*}
It follows that the \clsb{BPP}{q} machine accepts with probability at least $2\alpha(|x|)-1$ if $x \in L_{\sc{yes}}$, and at most $2\beta(|x|)-1$ if $x \in L_{\sc{no}}$. This completes the proof.

\section{Deterministic Classical Approximation Algorithms for the Guided Local Hamiltonian Problem}\label{app:classical-algorithms}
In the mild setting where the promise gap and overlap are constant, the \sc{Guided Local Hamiltonian} problem can be solved efficiently using classical algorithmic techniques.
Specifically, we demonstrate that the problem can be solved in polynomial time via the cluster expansion framework~\cite{mann2024algorithmic}.
The purpose of this section is to provide a straightforward proof that in the case of a constant-depth-preparable guiding state; the techniques of Ref.~\cite{mann2024algorithmic} can be applied to solve the problem efficiently.
Moreover, this proof simplifies the argument presented in Ref.~\cite{zhang2024dequantized}.
However, we are unable to extend this argument to the more general case of considering semi-classical subset states without enforcing a sequence of strong assumptions on the structure of the underlying Hamiltonian (see \cref{rmk:extension})

\subsection{Setup}
Let $G = (V,E)$ be a multihypergraph with maximum degree $\Delta$ and maximum rank $k$.
Let $H_G = \sum_{e \in E} \psi_e$ be a local Hamiltonian defined on $G$, where each $\psi_e$ is a positive semidefinite operator supported on the vertices in $e$.
We assume that $\norm{\psi_e} \leq 1$ for all $e \in E$; besides, this can be achieved by rescaling the Hamiltonian.
It is easy to see that $\norm{H_G} \leq \norm{G} \leq \poly{n}$, where $\norm{G} = \abs{E}$ is the number of edges in $G$ and $n = \abs{V}$ is the number of vertices.
Let the eigensystem of $H_G$ be $\{(\lambda_j, \ket{\phi_j})\}_{j=0}^{2^n-1}$, where $\ket{\phi_0}$ and $\lambda_0$ are the ground state and ground-state energy of $H_G$ respectively.
Denote the spectral gap of $H_G$ by $\gamma = \lambda_1 - \lambda_0 > 0$.
The partition function of $H_G$ at inverse temperature $\beta$ is defined as $Z_G(\beta) = \Tr[e^{-\beta H_G}]$.

Suppose $\ket{\xi} = U \ket{0^n}$ is a guiding state for $H_G$, where $U$ is a constant-depth quantum circuit with geometrically local gates (ensuring the transformed Hamiltonian remains local with constant degree).
Assume that the overlap between $\ket{\xi}$ and the ground state $\ket{\phi_0}$ is at least $\delta > 0$.
We define $p_j = \abs{\braket{\xi}{\phi_j}}^2$ for $j \in \{0,1,\dots,2^n-1\}$, then we have $p_0 \geq \delta$ and $\sum_{j=0}^{2^n-1} p_j = 1$.

The expectation value $E_G(\beta) = \mel{\xi}{\e^{-\beta H_G}}{\xi}$ is what we aim to approximate efficiently.
Let $\pi_0 = \bigotimes_{v\in V} \ketbra{0}{0}_v$ be the projector onto the all-zero state, and consider the projected partition function $Y^0_G(\beta) = \Tr[\pi_0 \e^{-\beta H_G}]$.
A straightforward calculation shows that $Y^0_{\tilde{G}}(\beta) = E_G(\beta)$, where $\tilde{G} = (V, \tilde{E})$ is a multihypergraph with the same vertex set as $G$ and edge set defined by the transformed Hamiltonian $H_{\tilde{G}} = U^\dagger H_G U$.
Simple light-cone arguments show that $H_{\tilde{G}}$ is also a local Hamiltonian defined on $\tilde{G}$ provided that $U$ is a constant-depth quantum circuit with geometrically local gates.
Specifically, if $G$ has maximum degree $\Delta$ and maximum rank $k$, then $\tilde{G}$ has maximum degree $\tilde{\Delta} = O(\Delta)$ and maximum rank $\tilde{k} = O(k)$ when $U$ has constant depth and bounded-range gates.
We can interpret $Y^0_{\tilde{G}}(\beta)$ as a thermal expectation value of the operator $\pi_0$ with respect to the Hamiltonian $H_{\tilde{G}}$ at inverse temperature $\beta$.
Standard results on classical approximation algorithms for partition functions of local Hamiltonians~\cite{mann2024algorithmic} can be applied to approximate $Y^0_{\tilde{G}}(\beta)$ efficiently, and hence $E_G(\beta)$.

\subsection{Cluster Expansion and Spectral Bounds}
\begin{theorem}[\cite{mann2024algorithmic}]\label{thm:mann-minko}
    Fix $\tilde{\Delta},\tilde{k} \in \mathbb{N}_{\geq 2}$.
    Let $\tilde{G} = (V,\tilde{E})$ be a multihypergraph with maximum degree $\tilde{\Delta}$ and maximum hyperedge rank $\tilde{k}$, and suppose $\beta$ is a complex number such that 
    \begin{equation*}
        \abs{\beta} \leq \frac{1}{\e^4 \tilde{\Delta} \binom{\tilde{k}}{2}}.
    \end{equation*}
    Then the cluster expansion for $\log Y_{\tilde{G}}^{0}(\beta)$ converges absolutely, $Y_{\tilde{G}}^{0}(\beta) \neq 0$, and there exists a classical algorithm that computes $\hat{z}$ satisfying
    \begin{equation*}
        \abs{\hat{z} - \log Y_{\tilde{G}}^{0}(\beta)} \leq \epsilon
    \end{equation*}
    in time polynomial in $n$ and $1/\epsilon$.
\end{theorem}

\begin{proposition}\label{prop:spectral-bounds}
    The spectral decomposition of $H_G$ implies that 
    \begin{equation}\label{eq:expectation-value-bound}
        \delta \e^{-\beta \lambda_0} \le E_G(\beta) \le (1+\delta/4) \e^{-\beta \lambda_0},
    \end{equation}
    provided $\e^{-\beta \gamma} \leq \delta/4$.
\end{proposition}
\begin{proof}
    From the spectral decomposition:
    \begin{equation*}
        E_G(\beta) = p_0 \e^{-\beta \lambda_0} + \sum_{j=1}^{2^n-1} p_j \e^{-\beta \lambda_j}.
    \end{equation*}
    The lower bound follows from $p_0 \geq \delta$. For the upper bound, using $\lambda_j \geq \lambda_0 + \gamma$ for $j \geq 1$:
    \begin{equation*}
        \sum_{j=1}^{2^n-1} p_j \e^{-\beta \lambda_j} 
        \leq (1-p_0) \e^{-\beta(\lambda_0 + \gamma)}
        \leq \e^{-\beta \lambda_0} \e^{-\beta \gamma}
        \leq \e^{-\beta \lambda_0} \cdot \frac{\delta}{4}.
    \end{equation*}
    Therefore $E_G(\beta) \leq (1 + \delta/4)\e^{-\beta \lambda_0}$.
\end{proof}

Let $\beta_{\text{CE}} = (\e^4 \tilde{\Delta} \binom{\tilde{k}}{2})^{-1}$ be a sufficient convergence radius of the cluster expansion and $\beta_{\text{gap}} = \gamma^{-1}\log(4\delta^{-1})$ be the inverse temperature at which the contribution from the excited states is sufficiently suppressed.
To apply both results, we require the existence of some $\beta$ satisfying $\beta \in [\beta_{\text{gap}}, \beta_{\text{CE}}]$, which is possible when
\begin{equation}\label{eq:gap-requirement}
    \gamma \geq \frac{\log(4\delta^{-1})}{\beta_{\text{CE}}} = \Omega(\tilde{\Delta} \tilde{k}^2 \log(\delta^{-1})).
\end{equation}
For constant $\delta$, $\tilde{\Delta}$, and $\tilde{k}$, this requires a spectral gap $\gamma = \Omega(1)$.

\begin{remark}\label{rem:ce-output}
    The cluster expansion (Theorem~\ref{thm:mann-minko}) provides a direct additive approximation to $\log Y_{\tilde{G}}^{0}(\beta)$.
    Specifically, for any $\epsilon > 0$, we can compute a value $\hat{z}$ such that $\abs{\hat{z} - \log Y_{\tilde{G}}^{0}(\beta)} \leq \epsilon$ in time polynomial in $n$ and $1/\epsilon$.
    Since $\log Y_{\tilde{G}}^{0}(\beta) = \log E_G(\beta)$, we have
    \begin{equation*}
        \log E_G(\beta) \in [\hat{z} - \epsilon, \hat{z} + \epsilon].
    \end{equation*}
    Taking logarithms on both sides of the bounds in~\cref{eq:expectation-value-bound}, we find
    \begin{equation}\label{eq:lambda-interval}
        \beta \lambda_0 \in [\log \delta - \hat{z} - \epsilon, \log(1+\delta/4) - \hat{z} + \epsilon].
    \end{equation}
    The width of this interval is $\log((1+\delta/4)/\delta) + 2\epsilon$.
\end{remark}

To decide the \sc{Guided Local Hamiltonian} problem, we use the estimator $\hat{\lambda} = -\frac{\hat{z}}{\beta}$.
From~\cref{eq:lambda-interval}, this satisfies
\begin{equation}\label{eq:estimator-bounds}
    \lambda_0 \in \left[
        \hat{\lambda} + \frac{\log \delta - \epsilon}{\beta}, 
        \hat{\lambda} + \frac{\log(1+\delta/4) + \epsilon}{\beta}
        \right].
\end{equation}

For the decision algorithm to work correctly with promise gap $b - a$, we need the uncertainty interval in~\cref{eq:estimator-bounds} to be smaller than the promise gap.
It is convenient to choose $\epsilon = \beta(b-a)/4$, which gives a total uncertainty of $\frac{b-a}{2} + \frac{\log((1+\delta/4)/\delta)}{\beta}$.
This is at most $b-a$ provided
\begin{equation}\label{eq:gap-condition}
    b - a \geq \frac{2}{\beta}\log\left(\frac{1+\delta/4}{\delta}\right).
\end{equation}

\subsection{Main Result}
\begin{theorem}\label{thm:glh-classical}
    Fix constants $\Delta, k$ and let $G = (V,E)$ be a multihypergraph with maximum degree $\Delta$ and maximum hyperedge rank $k$.
    Let $H_G$ be a local Hamiltonian defined on $G$ with spectral gap $\gamma>0$ and $\ket{\xi} = U \ket{0^n}$ be a guiding state for $H_G$ with overlap at least $\delta > 0$.
    Assume that $U$ is a constant-depth quantum circuit with geometrically local gates such that the transformed Hamiltonian $H_{\tilde{G}} = U^\dagger H_G U$ is a local Hamiltonian defined on a multihypergraph $\tilde{G} = (V, \tilde{E})$ with maximum degree $\tilde{\Delta} = O(\Delta)$ and maximum hyperedge rank $\tilde{k} = O(k)$.
    Take $a<b$ to be the thresholds in the \sc{Guided Local Hamiltonian} problem.
    Suppose the parameters satisfy:
    \begin{enumerate}
        \item $\delta = \Omega(1)$ (constant overlap),
        \item $\gamma = \Omega(\tilde{\Delta} \tilde{k}^2)$ (spectral gap condition from~\cref{eq:gap-requirement}),
        \item $b - a = \Omega(1)$ (constant promise gap satisfying~\cref{eq:gap-condition}).
    \end{enumerate}
    Then there exists $\beta = \Theta(1)$ such that 
    \begin{equation*}
        \gamma^{-1}\log(4\delta^{-1}) \leq \beta \leq (\e^4 \tilde{\Delta} \binom{\tilde{k}}{2})^{-1},
    \end{equation*}
    and there is a deterministic algorithm that decides the \sc{Guided Local Hamiltonian} problem for $H_G$ and $\ket{\xi}$ in time polynomial in $n$.
\end{theorem}

\begin{remark}\label{rmk:extension}
    The extension of this technique to the more general case of semi-classical subset states is non-trivial and would require strong assumptions on the structure of the underlying Hamiltonian.
    Specifically, when $\ket{\xi} = \sum_{x \in S} c_x \ket{x}$ is a semi-classical subset state, we aim to approximate $E_G(\beta) = \sum_{x,y \in S} Y^{x,y}_G(\beta)$, where $Y^{x,y}_G(\beta) = \Tr[\pi_{x,y} \e^{-\beta H_G}]$ and $\pi_{x,y} = \ketbra{x}{y}$.
    Each $Y^{x,y}_G(\beta)$ can be interpreted as a thermal expectation value of the operator $\pi_{x,y}$ with respect to the Hamiltonian $H_G$ at inverse temperature $\beta$.
    However, the operator $\pi_{x,y}$ is not positive semidefinite and there is no guarantee that the quantity $Y^{x,y}_G(\beta)$ is non-zero, which is a crucial requirement for the cluster expansion to be applicable.
    We expect that under a decay condition for basis states with large Hamming distance, the cluster expansion could be applied to approximate $Y^{x,y}_G(\beta)$ efficiently, but this is beyond the scope of this work.
\end{remark}
\end{document}